\documentclass[preprint, 10pt]{elsarticle}

\usepackage{amssymb}
\usepackage{amsmath, amsthm, mathtools, bm, subfig, scrextend, array, setspace, caption, comment, textcomp, soul, color}
\newcolumntype{M}[1]{>{\centering\arraybackslash}m{#1}}
\newtheorem{theorem}{Theorem}
\newtheorem{definition}{Definition}
\journal{Signal Processing Journal}

\begin{document}

\begin{frontmatter}



\title{Event Driven Sensor Fusion}


\author{Siddharth Roheda$^\star$, Hamid Krim$^\star$, Zhi-Quan Luo$^\dagger$, Tianfu Wu$^\star$\\
$^\star$ ECE Department, North Carolina State University \\
$^\dagger$ ECE Department, University of Minnesota}

%\affiliation{organization={},%Department and Organization
   %         addressline={}, 
      %%   postcode={}, 
            %state={},
           % country={}}

\begin{abstract}
	Multi sensor fusion has long been of interest in target detection and tracking. Different sensors are capable of observing different characteristics about a target, hence, providing additional information toward determining a target's identity. If used constructively, any additional information should have a positive impact on the performance of the system. In this paper, we consider such a scenario and present a principled approach toward ensuring constructive combination of the various sensors. We look at Decision Level Sensor Fusion under a different light wherein each sensor is said to make a decision on occurrence of certain events that it is capable of observing rather than making a decision on whether a certain target is present. These events are formalized to each sensor according to its potentially extracted attributes to define targets. The proposed technique also explores the extent of dependence between features/events being observed by the sensors, and hence generates more informed probability distributions over the events. In our case, we will study two different datasets. The first one, combines a Radar sensor with an optical sensor for detection of space debris, while the second one combines a seismic sensor with an acoustic sensor in order to detect human and vehicular targets in a field of interest. Provided some additional information about the features of the object, this fusion technique can outperform other existing decision level fusion approaches that may not take into account the relationship between different features. Furthermore, this paper also addresses the issue of coping with damaged sensors when using the model, by learning a hidden space between sensor modalities which can be exploited to safeguard detection performance.
\end{abstract}

%\includegraphics[width=\textwidth, height=0.5\textheight]{Graphical_Abstract.pdf}

%\item Targets of interest can be defined as combination of feature (velocity, weight) events.
%\item Fusion of information from multiple sources leads to improved classification.
%\item Important to explore extent of correlation between sensors when performing fusion.
%\item A common subspace exists between various modalities observing the same target.
%\item Such a common subspace can safeguard performance against sensor damage.


\end{frontmatter}


\section{Introduction}
\label{intro}
Often times more sensors are required in order to successfully detect and classify targets of interest. Additional sensors may provide supplementary information about a target, which can help the system make a more informed decision about its detection and classification. This data in turn often requires a degree of harnessing and fusion to seek an improved inference. Sensor fusion is generally known to broadly distinguish three levels of fusion, namely, data level, feature level, and decision level fusion. Data Level fusion generally processes raw data and performs fusion according to some criterion before making an inference. Feature level fusion, on the other hand, first gleans information from raw data (eg. transformed data) observed from diverse sensors, to subsequently coherently merge them for inference. In decision level fusion, each sensor reaches an individual decision, prior to optimal combination of these decisions to yield a more informed inference. The classical approach to decision level fusion is summarized in Figure \ref{DLF}. Over the years, we have seen classical techniques like Bayesian Fusion \cite{dempster2008generalization} and Dempster-Shafer Fusion\cite{shafer1976mathematical} used for combining sensors at the decision level. While more recently we have seen model based approaches \cite{thesis, li2001convex, florea2007critiques} that take into account the types of sensors that make up the network.
\begin{figure} 
	\centering
	\includegraphics[width = 0.69\textwidth]{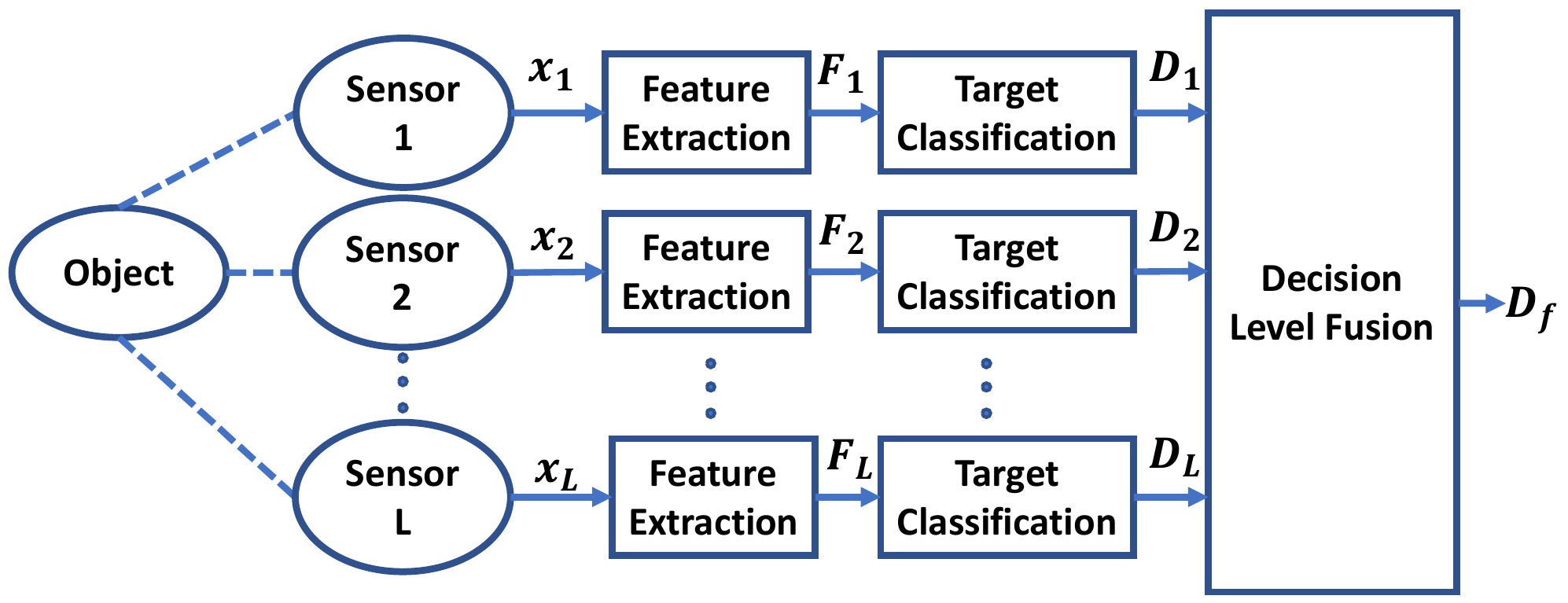}
	\centering
	\caption{Decision Level Fusion of multi-sensor observations}
	\label{DLF}
\end{figure}

\textbf{Our Contributions: } In this paper, we present a principled approach to decision level fusion for improved inference performance. A classification decision is reached by cataloging sets of events, along with the probabilistic characterization for each sensor, and following a joint probabilistic and coherent evaluation of these events. These events are formalized to each sensor according to its potentially extracted attributes to define targets. What this in effect achieves, is a probability measure assignment to a specific target following its description. Similar fusion inspired by feature events has been previously discussed in \cite{rohedaEDF}. Furthermore, we also address the practical situation where a sensor may be noisy or damaged and is no longer of use for fusion. We show that we can learn a hidden space between the sensors such that the fusion algorithm works around the damaged sensor, while achieving better performance than simply ignoring the damaged sensor, thus ensuring a graceful degradation. We formulate the problem of finding this hidden space by a criterion driven by the classification performance of a Support Vector Machine. In our case, we will study two different datasets. The first one, combines a Radar sensor with an optical sensor. A radar is used to explore the velocity of an object among other things, thus defining a sample space and a Sigma-Field with an associated probability measure, and is coupled to a telescopic sensor with an analogously associated probability space. This product space thus allows us to define a principled fusion framework with an improved and robust performance. Similarly, the second dataset will involve a seismic sensor, coupled with an acoustic sensor. 

\section{Related Work}
\label{rel_work}
As noted earlier, sensor fusion has long been of interest, albeit with limited theoretical success particularly when heterogeneous data are present, hence missing a unified and systematic approach which has remained elusive. An introduction and comprehensive survey to the area of fusion is provided in \cite{survey1,survey2}. As noted earlier, there has been significant research activity in decision level sensor fusion, starting from classical techniques like Bayesian Inference \cite{dempster2008generalization} and Dempster-Shafer Fusion \cite{shafer1976mathematical}. %The Bayesian Inference method is dependent on the knowledge of a-priori distributions on the targets of interest and the conditional probabilities of a certain sensor report occurring given a certain target. Bayesian Inference uses the Bayes' rule in order to fuse various 
Bayesian Fusion has shown success when prior knowledge about sensor reports is available. On the other hand, Dempster-Shafer fusion was proposed to specifically lift such a restriction on the information prior,  at a cost of a substantial increase in computational complexity. In \cite{huang1993behavior}, a two-stage approach to sensor fusion was proposed, involving knowledge-modeling, which learns from past behavior of classifiers whose results are to be fused, and operation stage, that combines outputs of these classifiers based on knowledge learned in the first stage. 
More recent work in decision level fusion is based on the sensor network model \cite{thesis}. Here, the network is modeled as either being made up of similar or dissimilar sensors. Similar Sensor Fusion \cite{thesis, li2001convex, florea2007critiques}, is used when all the sensors explore the same characteristics/features of the target (for example, a set of 5 radars, looking at the same target), while Dissimilar Sensor Fusion \cite{thesis,florea2007critiques} is alternatively used when sensors explore different characteristics/features of the target (for example, a radar and optical sensor looking
at the same target). In \cite{thesis}, Similar Sensor Fusion using Kulback-Liebler distance was shown to be equivalent to a weighted averaging of the predictions from each sensor. 
These assumptions turn out to be too restrictive, in that some sensors, albeit dissimilar, may have some common features while offering additional features to
enrich an object/target characterization. Our goal is to explore
such a case, and demonstrate that a systematic and principled
approach may be designed, and our resulting overall solution is improved
on account of this enhancement. 

\subsection*{Mutual Information}
Consider two random variables, X and Y, with a joint probability mass $ p(x,y) $ and marginal probability mass functions $ p(x) $ and $ p(y) $. The Mutual Information, $ I(X;Y) $, is the relative entropy between the joint distribution, $ p(x,y) $, and the product distribution, $ p(x)p(y) $ \cite{IT_text}. The formula for Mutual Information is then given as:

\begin{equation} \label{MI}
	I(X;Y) = \sum_{x \in X}\sum_{y \in Y} p(x,y)log\frac{p(x,y)}{p(x)p(y)}.
\end{equation}

Further, the relationship between mutual information and joint entropy of X and Y is given as \cite{IT_text}, 

\begin{equation} \label{MI-JE}
	H(X,Y) = H(X) +H(Y) - I(X;Y)
\end{equation}   

\section{Problem Formulation}\label{pf}
Assume throughout a set of targets/objects, $ O = \{o_1, o_2, ..., o_I\} $, whose detection and/or classification are of interest. Let the $ k^{th} $ feature observed by the $ l^{th} $ sensor be $ F_k^l $. Then, a set of mutually exclusive events, $ \Omega_k^l = \{\omega_{kj}^l\}_{j=1, ..., J_{kl}} $, may be defined for the feature $ F_k^l $. Here, $ \omega_{kj}^l $ is the $ j^{th} $ event for $ F_k^l $ and is described as, $ \omega_{kj}^l: F_k^l \in [u_j, v_j) $, $ u_j \in  {\rm I\!R^+}$, $ v_j \in {\rm I\!R^+} $, and $ v_j > u_j $. For example an event may be: The target in view is traveling at $velocity < 5 m/h$. The probability report for the $ k^{th} $ feature from the $ l^{th} $ sensor is then defined as
\begin{equation}
	D_k^l = \{\Omega_k^l, \sigma_B(\Omega_k^l), P_k^l\}. % \text{,   } l = 1, ..., L \text{, }k = 1, ..., K_l 
\end{equation}
Where, $ \sigma_B(\Omega_k^l) $ is the Borel sigma algebra of $ \Omega_k^l $, and can be thought of as the set of all possible events that can be described over the feature. $ P_k^l $  is the set of probabilities over the events in $ \sigma_B(\Omega_k^l) $.

Let the $n^{th}$ observation made by the $l^{th}$ sensor be denoted by $\bm{x^l_n} = \{x^l_{n_q}\}_{q=1,...,Q}$, where $q$ corresponds to the signal value at time $q$. Furthermore, let $C_{kj}^{l} (\bm{x^l_n}) = {\bm{w}_{kj}^{{l}^T}} \bm{x^l_n}$ be a scoring function that gives a detection score, where, ${\bm{w}_{kj}^{l}}$ is the weight vector for a classifier trained to detect the event $a_{kj}^l: F_k^l \in [u_j, v_j)$. The bias term for the classifier can be modeled by appending a constant feature to each $\bm{x_n^l}$. Then, the probability of occurrence of the corresponding event is determined as
\begin{equation}
	\label{Pkl}
	P_k^l(a_{kj}^l) = \frac{exp({C_{kj}^{l}(\bm{x^l_n})})}{\sum_m exp({C_{km}^{l}(\bm{x^l_n})})}.
\end{equation}

Since, we can only define objects by a set of characteristic features, it follows that a combination of certain events occurring over different features will be used working in the product space, 
\begin{equation}
	\Omega = \Omega_1^1 \times \Omega_2^1 \times ... \times {\Omega}_{K_{1}}^1 \times \Omega_1^2 \times \Omega_2^2 \times ... \times {\Omega}_{K_{2}}^2 \times \Omega_1^L \times \Omega_2^L \times ... \times {\Omega}_{K_{L}}^L
\end{equation}
where, $ K_l $ is the total number of features observed by the $ l^{th} $ sensor, and $ l = 1, ..., L $. Further, an object will be defined as some combination of events in this product space, $ o_i \in \sigma_B(\Omega)$. Given the object definitions and the probability distributions over various features, our goal is to then find the fused probability report over the objects, $ D_f = \{O, P_f\} $.

\section{Proposed Method}
The sensor reports form a set $\{D_k^l\}_{k=1,...,K_L}^{l=1,...,L}$ which potentially are different sensors providing different features making up events which define targets. Specifically, the definitions of objects are the result of algebraic operations on the event space $\sigma_B(\Omega)$, a Sigma-algebra on the product space , $\Omega$, with associated probability measures as noted in Section \ref{pf}. Thus, we must evaluate the probability distribution on $\sigma_B(\Omega)$.	
\subsection{Determining object Probabilities}
Consider the events, $\gamma_k^l \in \sigma_B(\Omega_k^l)$, and the corresponding product space, $\Omega$. Then, for any combination, $Comb(\gamma_k^l) \in \sigma_B(\Omega)$, the object probability may be determined as, $P_f(o) = g(Comb(\gamma_k^l))$, where $g$ is a function that uses rules of probability to determine the fused object probability. Considering a 2-D setting, an object may be defined as a combination of events $\gamma_1 \in \sigma_B(\Omega_1)$, and $\gamma_2 \in \sigma_B(\Omega_2)$. The combination defined in the product space, $\Omega = \Omega_1 \times \Omega_2$, may be of the form $o: \{\gamma_1 \wedge \gamma_2\}$ or $o: \{\gamma_1 \vee \gamma_2\}$. Given the joint probability $P_\Omega$, rules of probability can be used to determine the fused object probability as follows:
\begin{itemize}
	\item $o: \{\gamma_1 \wedge \gamma_2\}: P_f(o) = P_\Omega(\gamma_1, \gamma_2)$
	\item $o: \{\gamma_1 \vee \gamma_2\}: P_f(o) = P_1(\gamma_1) + P_2(\gamma_2) - P_\Omega(\gamma_1, \gamma_2)$
\end{itemize}
Where, $P_1(\gamma_1)$ and $P_2(\gamma_2)$ are the marginal probabilities for detection of the events $\gamma_1$ and $\gamma_2$ as seen by sensors $1$ and $2$.

This can be easily extended to any number of features and combinations of more than two events. 

\subsection{Determining the Joint Probability}
When determining the joint probability in the product space, $\Omega$, it is important to account for the extent of dependence between the features: Completely independent features yield minimal mutual information, and the joint distribution with the minimum mutual information should be selected; a high dependence between features, on the other hand, yields maximal mutual information, and the joint distribution with maximal mutual information should be selected. These are clearly the extreme cases of dependence, and do not address the partial dependence case. To account for partially dependent features, a good approximation to the joint probability would be a convex combination of the joint probabilities maximizing and minimizing the mutual information. For ease of writing, we use $\gamma_1^1,...\gamma_k^l,...,\gamma_{K_L}^L$ to represent $\gamma_1^1,\gamma_2^1...,\gamma_{K_1}^1, \gamma_1^2,\gamma_2^2 ...,\gamma_{K_2}^2,...,\gamma_1^L,\gamma_2^L...,\gamma_{K_L}^L$. The joint probability of events $\gamma_k^l \in \sigma_B(\Omega_k^l)$, can then be determined as,
\begin{equation} \label{convcomb}
	\begin{split}
		P_\Omega(\gamma_1^1,...,\gamma_k^l,...,\gamma_{K_L}^L) =  \rho.P_{\Omega_{\text{MAXMI}}}(\gamma_1^1,...,\gamma_k^l,...,\gamma_{K_L}^L) +\\ (1-\rho).P_{\Omega_{\text{MINMI}}}(\gamma_1^1,...,\gamma_k^l,...,\gamma_{K_L}^L), 
	\end{split}
\end{equation} 
where, $\rho \in [0,1]$ is a pseudo-measure of extent of correlation between the features. $\rho \approx 1$ when features are highly correlated, and $\rho = 0$ when features are independent of each other. $\rho$ can be estimated from the training data by either computing the correlation between the features by using a measure like Pearson's correlation/distance correlation \cite{szkely2009} or by optimizing $\rho$ over the training data. %Empirically, it is seen %(see Section \ref{exp}) that values of $\rho$ as determined by selecting best performance for training data, are close to selecting the distance correlation between the features.

It can be readily seen from Equation \ref{MI} that mutual information between two random variables is minimized when the joint probability distribution is selected as the product of the marginals, to yield,
%\begin{equation} 
%\begingroup\makeatletter\def\f@size{9}\check@mathfonts
\begin{equation}\label{minmi}
	%\begin{split} 
	P_{\Omega_{\text{MINMI}}}(\gamma_1^1,...,\gamma_k^l,...,\gamma_{K_L}^L) = \prod_{k,l} P_k^l(\gamma_k^l).
	%P_1^1(\gamma_1^1)P_2^1(\gamma_2^1)...P_k^l(\gamma_k^l)...\\P(\gamma_{K_L-1}^L)P(\gamma_{K_L}^L)
	%\end{split}
\end{equation} 
%\endgroup
%\end{equation}
Maximizing mutual information on the other hand, when given the marginal probabilities requires additional work. Given some random variables X and Y, and conditioning on the marginal probability distributions of X and Y yields constant $H(X) \text{ and } H(Y)$. As may be seen from Equation  \ref{MI-JE}, the maximization of Mutual Information between two random variables then becomes equivalent to minimizing their Joint Entropy, which is known to be a concave function. The special case of $P_{\Omega_{MAXMI}} = P_{\Omega_{MINMI}}$ does not impact Eq. (6) as the equation will result in $P_{\Omega} = \rho P_{\Omega_{MINMI}} + (1-\rho) P_{\Omega_{MAXMI}} = P_{\Omega_{MINMI}}$.
%\begin{gather}
%	\nonumber{P_{\Omega}}_{\text{MAX MI}} = \min_{P_\Omega} \sum_{x \in \Omega_k^l} \sum_{y \in \Omega_j^m} - P_\Omega (x, y) log P_\Omega(x, y)\\
%	\nonumber\text{subject to: } \sum_{x \in \Omega_k^l} {P_{\Omega}}(x,y) = P_j^m(y),\\	
%\sum_{y \in \Omega_j^m}	{P_{\Omega}}(x,y) = P_k^l(x),\text{ } {P_{\Omega}}(x,y) \geq 0.
%\end{gather}
\begingroup\makeatletter\def\f@size{8.5}\check@mathfonts
\begin{gather}
	\begin{align} 
		&\nonumber P_{\Omega_{\text{MAXMI}}} = \min_{P_\Omega}\text{ } \sum_{ \mathclap{\substack{b_1^1 \in \Omega_1^1,\\...,b_k^l \in \Omega_k^l,...,b_{K_L}^L \in \Omega_{K_L}^L}}}-P_\Omega(b_1^1,...,b_k^l,...,b_{K_L}^L)\log P_\Omega(b_1^1,...,b_k^l,...,b_{K_L}^L)\\
		&\nonumber \text{subject to: } \forall l \in \{1,...,L\}, \text{ } \forall k \in \{1,...,K_l\}, \\ 
		&\nonumber \quad \quad \quad \sum_{\mathclap{\{b_1^1,...,b_k^l,...,b_{K_L}^L\} \setminus b_k^l}} P_\Omega(b_1^1,...,b_k^l,...,b_{K_L}^L) = P_k^l(b_k^l),\\
		&\quad \quad \quad P_\Omega(b_1^1,...,b_k^l,...,b_{K_L}^L) \geq 0
	\end{align}
\end{gather}
\endgroup
%Where, $\Lambda$ is the set of features that we want to marginalize over.
\begin{figure*}[tbp]
	\centering
	\includegraphics[width = 0.95\textwidth]{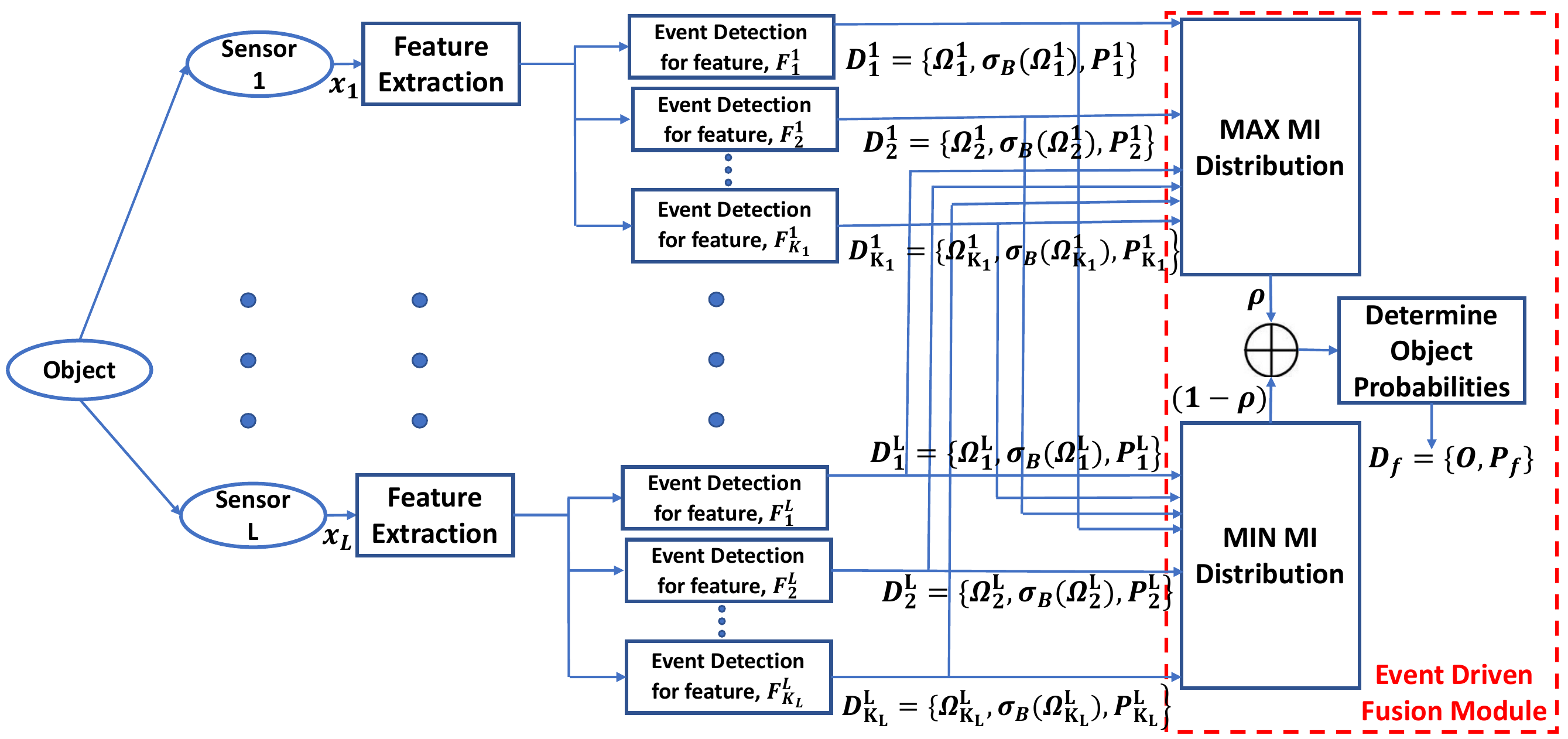}
	\centering
	\caption{Block diagram of the Event Driven Approach for fusion}
	\label{bd}
\end{figure*}
A greedy approach for minimizing joint entropy given the marginal probabilities can be constructed \cite{MAXMI} and is exploited to find the joint distribution with maximal mutual information. The main idea here is to keep large probability masses intact and not break them down into smaller chunks. The contribution of a probability mass toward the joint entropy only increases if it is divided into smaller chunks. That is, for $p = a+b$, $ -p.log(p) \leq -a.log(a) - b.log(b) $, when $ 0<p<1 $ and $ a,b > 0 $. So, keeping the large probability masses from given marginal probabilities intact ensures that their contribution towards the joint entropy is minimized. As empirically demonstrated in \cite{MAXMI}, the minimal joint entropies are obtained to within 1 bit of the optimal values. Figure \ref{bd} summarizes the steps of the proposed fusion approach in a block diagram.

\subsection{Robustness: Addressing damaged sensors}
\label{damaged_sensors}
%Often, sensors may get damaged during the implementation phase, and may no longer be useful for fusion. Generally fusion algorithms ignore the damaged sensors in order to continue producing coherent decisions. This means they are no longer using the information residing in the training data of the damaged sensor either. and sometimes there are limitations on the types of sensors that can be deployed
In practice, sensor measurements may often be noisy, missing, or unusable in unconstrained surveillance settings, or just of limited capacity. In this scenario, it is common to ignore such sensors, with a potentially negative impact on optimal performance (i.e. all sensors are available and functional). We consider exploiting prior knowledge about the relationship between the various modalities, so that our system can safeguard a high detection accuracy. This prior knowledge resides in the training data, which is assumed to be available for all the modalities. Such a problem has previously been studied in \cite{roheda2018cross}, where Conditional Generative Adversarial Networks (CGAN) were used to replicate features of damaged sensors. This requires that the features for optimal classification be known before hand, so that the CGAN network can learn to replicate them, whereas, in our case we are searching a hidden space that is shared between sensor modalities, even with the absence of the optimal features. To that end, we propose to find linear operators that transform each sensor modality into a common hidden space\footnote{In \cite{wang2018fusing}, such a space was referred to as an `information subspace'}, so that it represents the shared information between the sensors. It is also important that we formulate the cost functional so that the determined hidden space is discriminative with respect to detection of event occurrences. We discuss two approaches to find such a hidden space. The Global Hidden Space (GHS) approach searches for a single common representation for all the features of interest, and transforms each sensor observation into this common space. On the other hand, the Independent Hidden Space (IHS) approach searches for an independent representation for each feature of interest. While IHS reduces the number of constraints on each hidden space, it increases the number of parameters as a transformation to each of these independent hidden spaces must be learnt.
%The first approach finds a global hidden space representing all features ($\{F_k^l\}$), while the second approach finds an independent hidden space for each feature. 

\subsubsection{Global Hidden Space}
\label{GlobalHS}
\begin{figure*}
	%\centering
	\includegraphics[width = 0.9\textwidth]{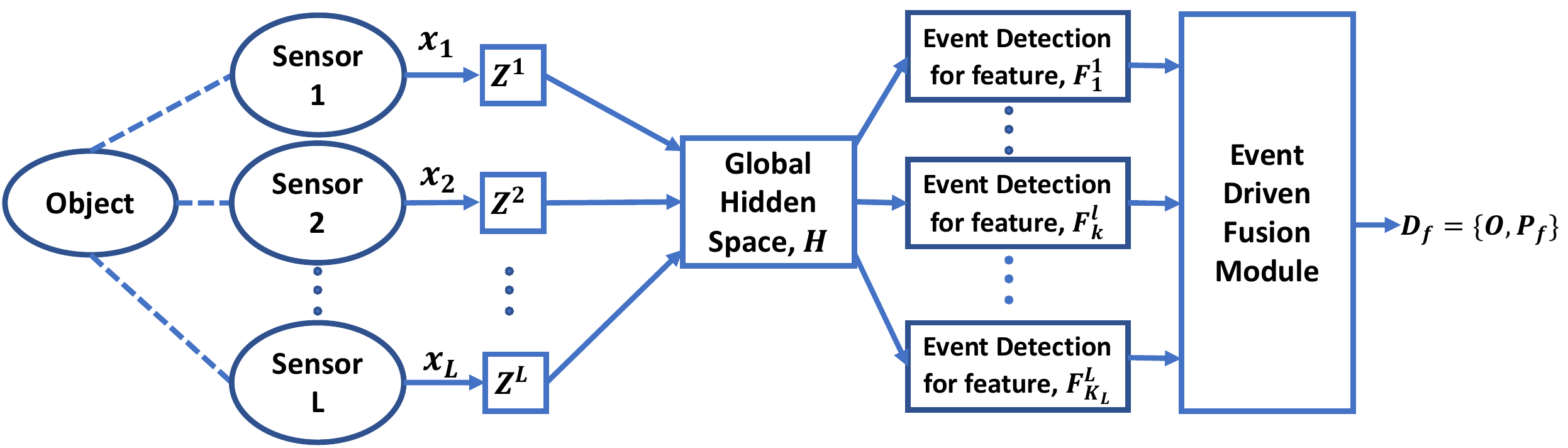}
	\centering
	\caption{Using a Global Hidden Space for Event Driven Fusion}
	\label{GHS}
\end{figure*}
We assume here the existence of a global space that can characterize all the features of interest, i.e. the same hidden space may be used to detect all feature events, as shown in Figure \ref{GHS}. We seek to find the linear operators, $\bm{Z^1, Z^2,...,Z^L}$, such that, 
\begin{equation}
	\label{GlobalTx}
	\forall l \in \{1,...,L\},\text{ }  \bm{Z}^{\bm{l}}_{d \times d_l} {\bm{X}^{\bm{l}}_{d_l \times N}} = \bm{H}_{d \times N}.  %F_1\bm{X_1} = F^2\bm{X_2} = ... = F^L\bm{X_L} = H
\end{equation} 
Where, $\bm{X^l} = \{\bm{x^l_{n}}\}_{n=1,...,N}$, and $N$ is the number of training samples. The desired dimension of the hidden space is denoted by $d$, while the dimension of the observed signal from the $l^{th}$ sensor, by $d_l$ ($d < d_l, \forall l$). 

A key observation to our goal of determining a common subspace for different modalities, is that if a set of linear operators commute, they share common eigenvectors \cite{commoneigbasis, Frobenius_Commutation}. If these operators are furthermore individually diagonalizable, they will share all their eigenvectors, leading to a common eigenbasis/subspace. 
	\begin{definition}
	Linear operators $\bm{A} \in {\rm I\!R^{nxn}}$ and $\bm{B} \in {\rm I\!R^{nxn}}$ are said to commute if, 
	\begin{equation}
		[\bm{A}, \bm{B}] = \bm{A}\bm{B} - \bm{B}\bm{A} = 0.
	\end{equation}
\end{definition}

\begin{theorem} \label{commutation1}
	If $\bm{A} \in {\rm I\!R^{nxn}}$ and $\bm{B} \in {\rm I\!R^{nxn}}$ are commuting linear operators, they share common eigenvectors.
\end{theorem}
\begin{proof}
	Consider an eigenbasis $\mathcal{V} = \{v_i\}$, of $\bm{A}$ with $\lambda_i$ the eigenvalue associated to $v_i$. Then for any $v_i$, 
	\begin{equation}
		\bm{AB}v_i = \bm{BA}v_i = \lambda_i\bm{B}v_i,
	\end{equation}
	i.e., if $\bm{B}v_i \neq 0$, $\bm{B}v_i$ is an eigenvector of A, associated to the same eigenvalue as $v_i$, $\lambda_i$. 
\end{proof}
\begin{theorem}\label{commutation2}
	If $\bm{A} \in {\rm I\!R^{nxn}}$ and $\bm{B} \in {\rm I\!R^{nxn}}$ are commuting operators that are also individually diagonalizable, they share a common eigenbasis.
\end{theorem}
\begin{proof}
	If $\bm{A}$ and $\bm{B}$ are individually diagonalizable, they have $n$-distinct eigenvalues, i.e. $\bm{A}$ can be diagonalized as, $\bm{A} = \bm{P}\bm{D_A}\bm{P}^{-1}$, where, $\bm{D_A}$ is an $n \times n$ diagonal matrix with eigenvalues of $\bm{A}$ on the diagonal, and $\bm{P}$ is an $n \times n$ matrix with eigenvectors of $\bm{A}$ as columns. Since, both $\bm{A}$ and $\bm{B}$ share common eigenvectors (as seen in Theorem \ref{commutation1}), $\bm{B}$ can also be diagonalized as $\bm{B} = \bm{P}\bm{D_B}\bm{P}^{-1}$. Hence, $\bm{A}$ and $\bm{B}$ share a common eigenbasis.   
\end{proof}
As a result, if the operators $\bm{Z^1,...,Z^L}$ commute and are diagonalizable, they will share a common eigenbasis. Furthermore, since the transformation $ \bm{Z^lX^l}$ lies in the range space of the linear operator, $\bm{Z^l}$, $\bm{Z^lX^l}$ lie in a common subspace, $\forall l \in \{1,...,L\},$ due to the shared basis. This hence yields a common feature representation for the different modalities.

To ensure pairwise commutation between the linear operators, $\{\bm{Z^l}\}^{l \in \{1,...,L\}}$, we must make all the operator matrices square, which can in turn be accomplished by using sampled random matrices, $\{\bm{U^l}\}^{l \in \{1,...,L\}}$. Since, $\bm{U^l}$ is a random projection which will stay constant during the learning process, the information about the transformation, $(\bm{Z^lU^l})\bm{X^l}$, still lies in the range of $\bm{Z^l}$. This results in Equation \ref{GlobalTx} re-expressed as,
\begin{equation}
	\label{GlobalTxRand}
	\forall l \in \{1,...,L\},\text{ }  \bm{Z}^{\bm{l}}_{d \times d} \bm{U^l}_{d \times d_l} {\bm{X}^{\bm{l}}_{d_l \times N}} = \bm{H}_{d \times N}.
\end{equation} 
To optimize the event detection on the basis of this hidden feature space we proceed to train a classifier for event detection, and to learn the operators $\bm{Z^l}$. Let, $\bm{W_k^l} = \{w_{kj}^l\}_{j=1,...,J_{K_l}}$, be the weight matrix for classification of events $\Omega_k^l = \{\omega_{kj}^l\}_{j=1,...,J_{K_l}}$ defined over the $k^{th}$ feature from the $l^{th}$ sensor. We build on the SVM formulation \cite{crammer-singer} to uncover the optimal hidden space with sufficient information to successfully detect all events over all features. To that end, the operators, $\bm{Z^1,...,Z^L}$ are sought by an optimization of an energy cost functional which includes a penalty term to encourage pairwise commutation of their application on the various sensor data. The objective then becomes, %Furthermore, allowing some error in commutation allows us to take into account that although the modalities share some information, they are not exactly the same. 
%Pairwise commutation is required as commutation is not inductive.
\begingroup\makeatletter\def\f@size{9}\check@mathfonts
\begin{gather} \label{ghs_opt}
	\begin{align}
		&\nonumber \min_{\bm{W_{k}^l}, \bm{\xi_k^l}, \bm{Z^l}} J(\bm{W_{k}^l}, \bm{\xi_{k}^l}, \bm{Z^l}) = \frac{1}{2} \sum_{l=1}^L \sum_{k=1}^{K_l} ||\bm{W_k^l}||^2 + C_1 \sum_{l=1}^L \sum_{n=1}^N \xi_{k_n}^{l}\\ 
		&\nonumber + \frac{1}{2} \sum_{\substack{l,m=1 \\ l\neq m}}^L (C_2||[\bm{Z^l}, \bm{Z^m}]||^2 + C_3 \sum_{n=1}^N (\bm{Z^l} \bm{U^l} \bm{x^l_n} - \bm{Z^m} \bm{U^m} \bm{x^m_n})^2)\\
		&\nonumber \text{subject to: } \\
		&\nonumber \forall l \in \{1,...,L\}, \forall k \in \{ 1,...,K_l\},\text{ }\\
		&\nonumber w_{k_{y_n}}^{l^T} (\bm{Z^l U^l x^l_n}) -  w_{k_{t}}^{l^T} (\bm{Z^l U^l x^l_n}) \geq 1 - \xi_{k_n}^l,\\ 
		&\nonumber \quad \quad \quad \quad \quad \quad \quad \quad \quad \quad \quad \quad \quad \quad \quad t \in \{1,...,J_{K_l}\} \setminus \{y_n\},\\
		& \xi_{k_n}^l \geq 0.
	\end{align}  
\end{gather}
\endgroup
The above cost functional does not guarantee the individual diagonalizability of the operators, $\bm{Z^1,...,Z^l}$, we, however, empirically observe that resulting operators on convergence are diagonalizable in most cases. The quadratic error constraint in Equation \ref{ghs_opt} encourages corresponding projected samples from different modalities to be close.
	\begin{figure*} 
	\centering
	\includegraphics[width = 0.97\textwidth]{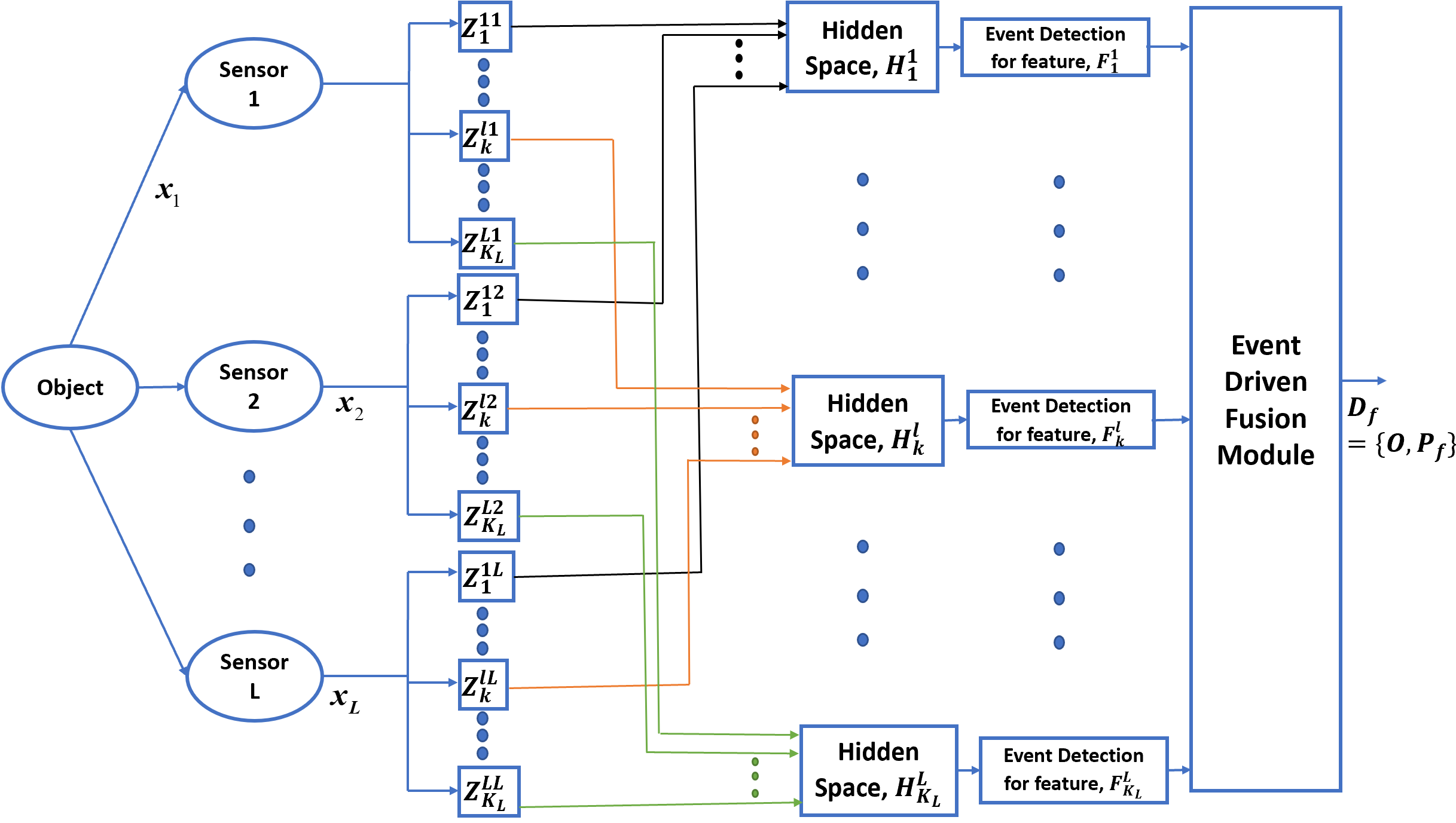}
	\centering
	\caption{Using Independent Hidden Spaces for Event Driven Fusion}
	\label{IHS}
\end{figure*}
\subsubsection{Independent Hidden Spaces}
\label{IndependentHS}
To reduce the computational complexity due to the number of constraints, we seek to find independent hidden spaces for the features of interest, as illustrated in Figure \ref{IHS}. We thus seek linear operators $\bm{Z_1^{11},...,Z_k^{lr},...,Z_{k_L}^{LL}}$, such that,
\begin{equation}
	\begin{split}
		\forall l,r \in \{1,...,L\}, \forall k \in \{1,...,K_l\}, \quad \bm{Z}_{\bm{k}_{d \times d_r}}^{\bm{lr}} \bm{X}^{\bm{r}}_{d_r \times N} = \bm{H}_{{\bm{k}}_{d \times N}}^{\bm{l}},
	\end{split} 
\end{equation}
%\begin{comment}
%\end{comment}
where, $\bm{Z_k^{lr}}$ transforms observed data, $\bm{X^r}$, from the $r^{th}$ sensor to the hidden space, $\bm{H_k^l}$.
This means, we now have a individual hidden space describing each feature of interest ($F_k^l$), making the number of cost functions that are independently optimized equivalent to the total number of features of interest $\sum_l K_l$. Furthermore, as in Section \ref{GlobalHS}, we again introduce the randomly sampled matrices, $\{\bm{U_k^{lr}}\}_{k \in \{1,...,K_l\}}^{l,r \in \{1,...,L\}} $, in order to help enforce pairwise commutation between the transformations,
\begin{equation}
	\label{IHS_formuln}
	\begin{split}
		\bm{Z}_{\bm{k}_{d \times d}}^{\bm{lr}} \bm{U}_{\bm{k}_{d \times d_r}}^{\bm{lr}} \bm{X}^{\bm{r}}_{d_r \times N} = \bm{H}_{\bm{k}_{d \times N}}^{\bm{l}},
		\forall l,r \in \{1,...,L\}, \forall k \in \{1,...,K_l\}.
	\end{split} 
\end{equation}
The solution in Equation \ref{IHS_formuln}, i.e. the hidden space for the $k^{th}$ feature from the $l^{th}$ sensor $\bm{H_k^l} = \bm{Z_k^{lr} U_k^{lr}X^r}$, is obtained by minimizing the following objective,
\begin{gather}
	\label{ihs_opt}
	\begin{align}	
		&\nonumber \min_{\bm{W_{k}^l},\bm{ \xi_k^{lr}}, \bm{Z_k^{lr}}} J(\bm{W_k^l}, \bm{\xi_k^{lr}}, \bm{Z_k^{lr}}) = \frac{1}{2} ||\bm{W_k^l}||^2 + C_1\sum_{r=1}^L \sum_{n=1}^N \xi_{k_n}^{lr} \\
		&\nonumber \quad \quad + \frac{1}{2}\sum_{\substack{r,s=1\\ r\neq s}}^L (C_2||[\bm{Z_k^{lr}}, \bm{Z_k^{ls}}]||^2 + C_3 \sum_{n=1}^N (\bm{Z_k^{lr}U_k^{lr}x^r_n} - \bm{Z_k^{ls}U_k^{ls}x^s_n})^2)\\
		&\nonumber \text{subject to: }\\
		&\nonumber \forall r \in\{1,...,L\},\\
		& \nonumber \quad w_{k_{y_n}}^{l^T}(\bm{Z_k^{lr}U_k^{lr}x^r_n}) - w_{k_t}^{l^T}(\bm{Z_k^{lr}U_k^{lr}x^r_n}) \geq 1-\xi_{k_n}^{lr}, \quad  t \in \{1,...,J_{K_l}\} \setminus \{y_n\},\\
		& \quad \xi_{k_n}^{lr} \geq 0.
	\end{align}
\end{gather}
The above conditions are satisfied by setting,
\begin{equation} \label{xi_n}
	\xi_{k_n}^{lr} = \max_{t \neq y_n} \{0, 1 - w_{k_{y_n}}^{l^T}(\bm{Z_k^{lr}U_k^{lr}x_n^r}) + w_{k_t}^{l^T}(\bm{Z_k^{lr}U_k^{lr}x_n^r}) \},
\end{equation} 
%\begingroup\makeatletter\def\f@size{8.5}\check@mathfonts
%\begin{gather}
%\nonumber \min_{\bm{W_k^l}, \bm{Z_k^{lr}}} J = \frac{1}{2} ||\bm{W_k^l}||^2 + C_1 \sum_{r,n} \max_{t \neq y_n} \{0, 1 - w_{k_{y_n}}^{l^T}(\bm{Z_k^{lr}U_k^{lr}x_n^r}) + w_{k_t}^{l^T}(\bm{Z_k^{lr}U_k^{lr}x_n^r})\}\\
%\quad + \frac{C_2}{2} \sum_{\substack{r,s=1 \\ r \neq s}}^L (||[\bm{Z_k^{lr}}, \bm{Z_k^{ls}}]||^2 + \sum_{n=1}^N (\bm{Z_k^{lr}U_k^{lr}x^r_n} - \bm{Z_k^{ls}U_k^{ls}x^s_n})^2)
%\end{gather}
%\endgroup
Using $\xi_{k_n}^{lr}$ from Equation \ref{xi_n} into Equation \ref{ihs_opt} we get,
\begingroup\makeatletter\def\f@size{8.5}\check@mathfonts
\begin{gather}
	\begin{align}
		&\nonumber \min_{\bm{W_k^l}, \bm{Z_k^{lr}}} J(\bm{W_k^{l}}, \bm{Z_k^{lr}}) = \frac{1}{2} ||\bm{W_k^l}||^2 + C_1 \sum_{r,n} \max_{t \neq y_n} \{0, 1 - w_{k_{y_n}}^{l^T}(\bm{Z_k^{lr}U_k^{lr}x_n^r}) + w_{k_t}^{l^T}(\bm{Z_k^{lr}U_k^{lr}x_n^r})\}\\
		&\quad + \frac{1}{2} \sum_{\substack{r,s=1 \\ r \neq s}}^L (C_2||[\bm{Z_k^{lr}}, \bm{Z_k^{ls}}]||^2 + C_3\sum_{n=1}^N (\bm{Z_k^{lr}U_k^{lr}x^r_n} - \bm{Z_k^{ls}U_k^{ls}x^s_n})^2)
	\end{align}
\end{gather}
\endgroup
In order to update the variables, $\bm{W_k^l}$ and $\bm{Z_k^{lr}}$, a derivative of $J(\bm{W_k^l}, \bm{Z_k^{lr}})$ must be computed with respect to each variable. Let $V_{t_n} = 1 - w_{k_{y_n}}^{l^T}(\bm{Z_k^{lr}U_k^{lr}x_n^r}) + w_{k_t}^{l^T}(\bm{Z_k^{lr}U_k^{lr}x_n^r})$, $\forall t \in \{1,...,J_{K_l}\} \setminus y_n$ and $t_n' = \text{argmax}_{t \neq y_n} V_{t_n}$. Then we have, $\forall m \in \{1,...,J_{K_l}\}$,
\begingroup\makeatletter\def\f@size{9}\check@mathfonts
\begin{equation}
	\begin{split}
		\frac{dJ(\bm{W_k^l}, \bm{Z_k^{lr}})}{dw_{k_{m}}^l} = w_{k_{m}}^l + \mathcal{I}(V_{t'_n} > 0).C_1 \sum_{r,n} [\mathcal{I}(m=t)\bm{Z_k^{lr} U_k^{lr} x_n^r}\\ - \mathcal{I}(m={y_n})\bm{Z_k^{lr} U_k^{lr} x_n^r}]
		%\begin{cases}
		%\begin{split}
		%	w_{k_{m}}^l + C_1 \sum_{r,n} [\mathcal{I}(m=t)\bm{Z_k^{lr} U_k^{lr} x_n^r}\\ - \mathcal{I}(m={y_n})\bm{Z_k^{lr} U_k^{lr} x_n^r}], \\\text{ if } 1 - V_{t'} > 0
		%\end{split}\\
		%\begin{split}
		%w_{k_{m}}^l + C_1 \sum_{r,n} \bm{Z_k^{lr} U_k^{lr} x_n^r}, \text{ } 1 - V_{t'} > 0 \\ \text{ and } m = t'
		%\end{split}\\
		%	w_{k_{y_n}}^l,\text{ otherwise}\\ 
		%	\end{cases}
	\end{split}
\end{equation}
\endgroup
$\forall r \in \{1,...,L\}$,
\begingroup\makeatletter\def\f@size{8}\check@mathfonts
\begin{equation}
	\begin{split} 
		\frac{dJ(\bm{W_k^l}, \bm{Z_k^{lr}})}{d\bm{Z_k^{lr}}} = \mathcal{I}(V_{t'} > 0).C_1 \sum_{n} (- w_{k_{y_n}}^{l^T} (\bm{U_k^{lr} x_n^r}) + w_{k_{t'}}^{l^T} (\bm{U_k^{lr} x_n^r}))\\ + \sum_{s} [C_2([\bm{Z_k^{lr}, Z_k^{ls}}]\bm{Z_k}^{\bm{ls}^T} - \bm{Z_k^{ls}}[\bm{Z_k^{lr}, Z_k^{ls}}]) \\ + C_3(\sum_n \bm{U_k^{lr} x_n^r}(\bm{Z_k^{lr} U_k^{lr} x_n^r} - \bm{Z_k^{ls} U_k^{ls} x_n^s}))]
	\end{split}
	%\begin{cases}
	%\begin{split}
	%C_1 \sum_{n} (- w_{k_{y_n}^{l^T}} (\bm{U_k^{lr} x_n^r}) + w_{k_{t'}}^{l^T} (\bm{U_k^{lr} x_n^r}))\\ + \sum_{s} [C_2([\bm{Z_k^{lr}, Z_k^{ls}}]\bm{Z_k}^{\bm{ls}^T} - \bm{Z_k^{ls}}[\bm{Z_k^{lr}, Z_k^{ls}}]) \\ + C_3(\sum_n \bm{U_k^{lr} x_n^r}(\bm{Z_k^{lr} U_k^{lr} x_n^r} - \bm{Z_k^{ls} U_k^{ls} x_n^s}))], \\
	%\text{ if }1 - V_{t'} > 0
	%\end{split}\\
	%           \\
	%\begin{split}
	%\sum_s [C_2([\bm{Z_k^{lr}, Z_k^{ls}}]\bm{Z_k}^{\bm{ls}^T} - \bm{Z_k^{ls}}[\bm{Z_k^{lr}, Z_k^{ls}}]) \\ + C_3(\sum_n \bm{U_k^{lr} x_n^r}(\bm{Z_k^{lr} U_k^{lr} x_n^r} - \bm{Z_k^{ls} U_k^{ls} x_n^s}))], \\
	%\text{otherwise}
	%\end{split}  
	%\end{cases} 
\end{equation}
\endgroup
where, $\mathcal{I}$ is the indicator function,
\begin{equation}
	\mathcal{I}(a) =
	\begin{cases}
		1, \text{ if } \text{a is true}\\
		0, \text{ otherwise}.
	\end{cases}
\end{equation}
These derivatives are then used to update the variables at each iteration,
\begin{gather}
	\begin{align}
		&\nonumber \forall r \in \{1,...,L\}, \\
		&w_{k_m}^{l^{\{i+1\}}} = w_{k_m}^{l^{\{i\}}} - \mu \frac{dJ(\bm{W_k^l}, \bm{Z_k^{lr}})}{dw_{k_m}^l},\\ %(W_k^l, Z_k^{lr})
		&\bm{Z_k}^{{\bm{lr}}^{\{i+1\}}} = \bm{Z_k}^{{\bm{lr}}^{\{i\}}} - \mu \frac{dJ(\bm{W_k^l}, \bm{Z_k^{lr}})}{d\bm{Z_k^{lr}}}.
	\end{align}
\end{gather}
%	whose substitution in Equation \ref{ihs_opt} leads gradient descent to a solution by the way of the algorithm, namely the optimal $\bm{W_k^l}$ and $\bm{Z_k^{lr}}, \forall r \in \{1,...,L\}$,
%	\begin{gather}
%	\begin{align}
%	&\nonumber \forall m \in \{1,...,J_{K_l}\},\\
%	&w_{k_m}^{l^{\{i+1\}}} = w_{k_m}^{l^{\{i\}}} - \mu \frac{dJ(\bm{W_k^l}, \bm{Z_k^{lr}})}{dw_{k_m}^l},\\ %(W_k^l, Z_k^{lr})
%	&\bm{Z_k}^{{\bm{lr}}^{\{i+1\}}} = \bm{Z_k}^{{\bm{lr}}^{\{i\}}} - \mu \frac{dJ(\bm{W_k^l}, \bm{Z_k^{lr}})}{d\bm{Z_k^{lr}}}, 
%	\end{align}
%	\end{gather}
where, $i$ denotes the iteration number, and $\mu$ is the learning rate. %See appendix for the derivations.

If the $m^{th}$ sensor is damaged, the hidden spaces for this sensor are recovered from the available set of sensors, $\Gamma = \{1,...,L\} \setminus m$,
\begin{equation}
	\forall k \in \{1,...,K_m\}, \bm{H_k^m} = \frac{\sum_{r \in \Gamma} \bm{Z_k^{mr} U_k^{mr} X^r}}{|\Gamma|},
\end{equation}
where, $|\Gamma|$ is the cardinality of $\Gamma$.  
Figure \ref{Independent_HS_recovery} illustrates a scenario with two sensors, $l\in \{1,2\}$, and demonstrates the recovery of independent hidden spaces for features defined for a damaged sensor $(l=1)$, using observed data of the available sensor $(l=2)$.
\begin{figure}[h!] 
	\centering
	\includegraphics[width = 0.75\textwidth]{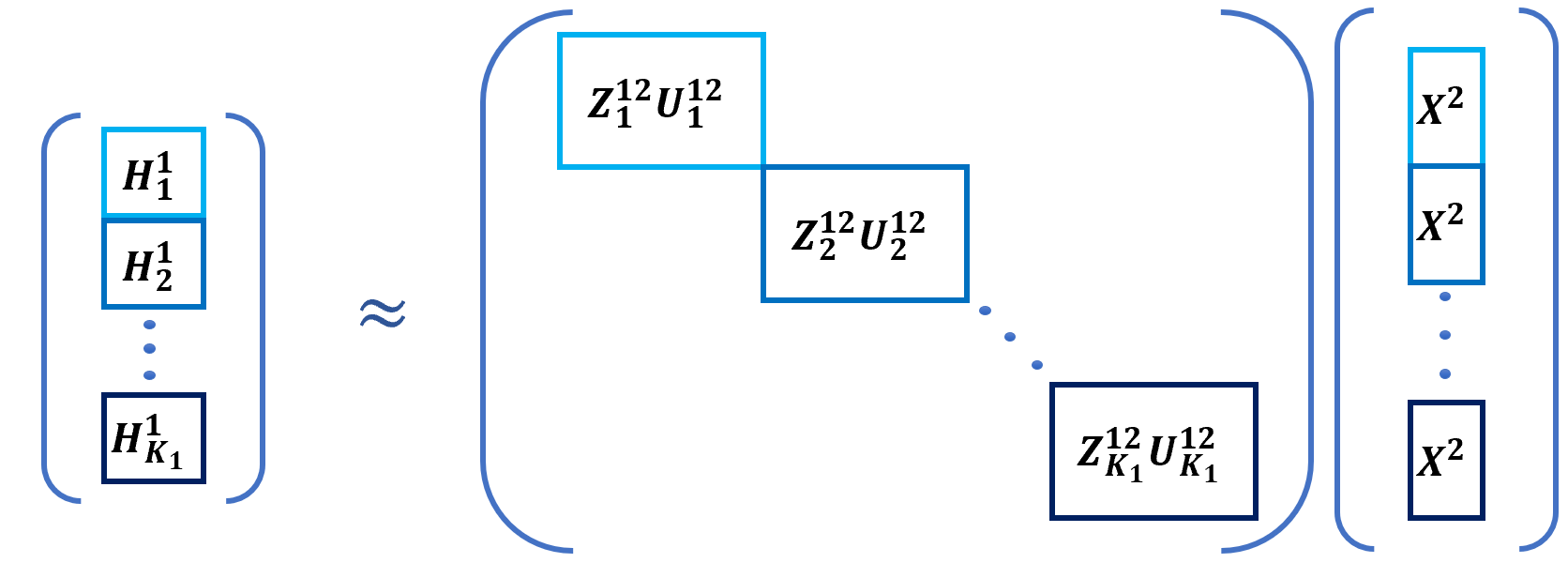}
	\centering
	\caption{Hidden Space recovery for a two sensor scenario, $l\in \{1,2\}$.  hidden space for features defined for damaged sensor, $l=1$, are recovered by observations from sensor, $l=2$}
	\label{Independent_HS_recovery}
\end{figure}

\textbf{Example:}
To visualize the result of this algorithm at convergence, consider a toy example with two modalities, $\bm{X^1} \in {\rm I\!R^{4 \times N}}$ and $\bm{X^2} \in {\rm I\!R^{3 \times N}}$, for binary classification. The random projections $\bm{U^1} \in {\rm I\!R^{2 \times 4}}$ and $\bm{U^2} \in {\rm I\!R^{2 \times 3}}$ are first used to project the data onto a 2-dimensional space, as seen in Figure \ref{toy_eg}-(a). Following this, we use the proposed approach to find hidden spaces, and project each modality onto a common hidden space, $\bm{H}_{2 \times N} \approx \bm{Z^1}_{2 \times 2} \bm{U^1}_{2 \times 4} \bm{X^1}_{4 \times N} \approx \bm{Z^2}_{2 \times 2} \bm{U^2}_{2 \times 3} \bm{X^2}_{3 \times N}$. Figures \ref{toy_eg}-(b),(c) show the data transformation into the hidden space for two cases: 1) When no penalty was enforced for non-commuting operators (Figure \ref{toy_eg}-(b)), and 2) When commutation between $\bm{Z^1} \text{ and } \bm{Z^2}$ was enforced (Figure \ref{toy_eg}-(c)). As may be seen from the determined hidden space in both cases, commutation is able to push the determined hidden subspace to be common for both modalities. Furthermore, it can be seen that the common classifier (denoted by the black solid line), which is learned jointly with the linear operators, is a compromise between the optimal classifiers for each modality (i.e. the SVM classifier learned for each modality individually in the transformed space). 
\begin{figure}
	\centering
	\subfloat[Random projection of $\bm{X^1} \text{ and } \bm{X^2}$ into a 2-d space]{\includegraphics[width = 0.33\textwidth]{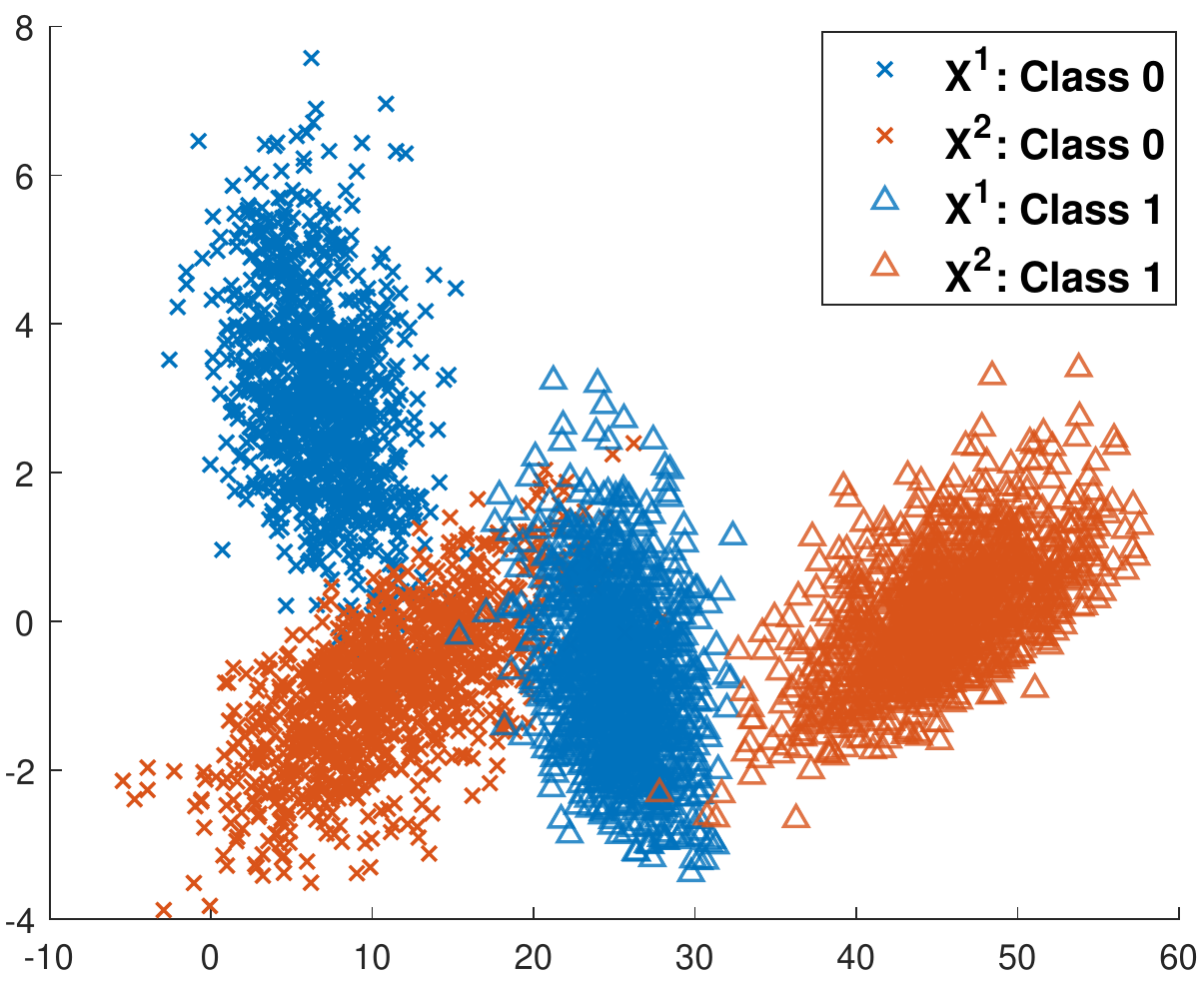}} \hfill
	\subfloat[Transformations $\bm{Z^1U^1X^1}$ and $\bm{Z^2U^2X^2}$ as determined without enforcing commutation between $\bm{Z^1} \text{ and } \bm{Z^2}$]{\includegraphics[width = 0.33\textwidth]{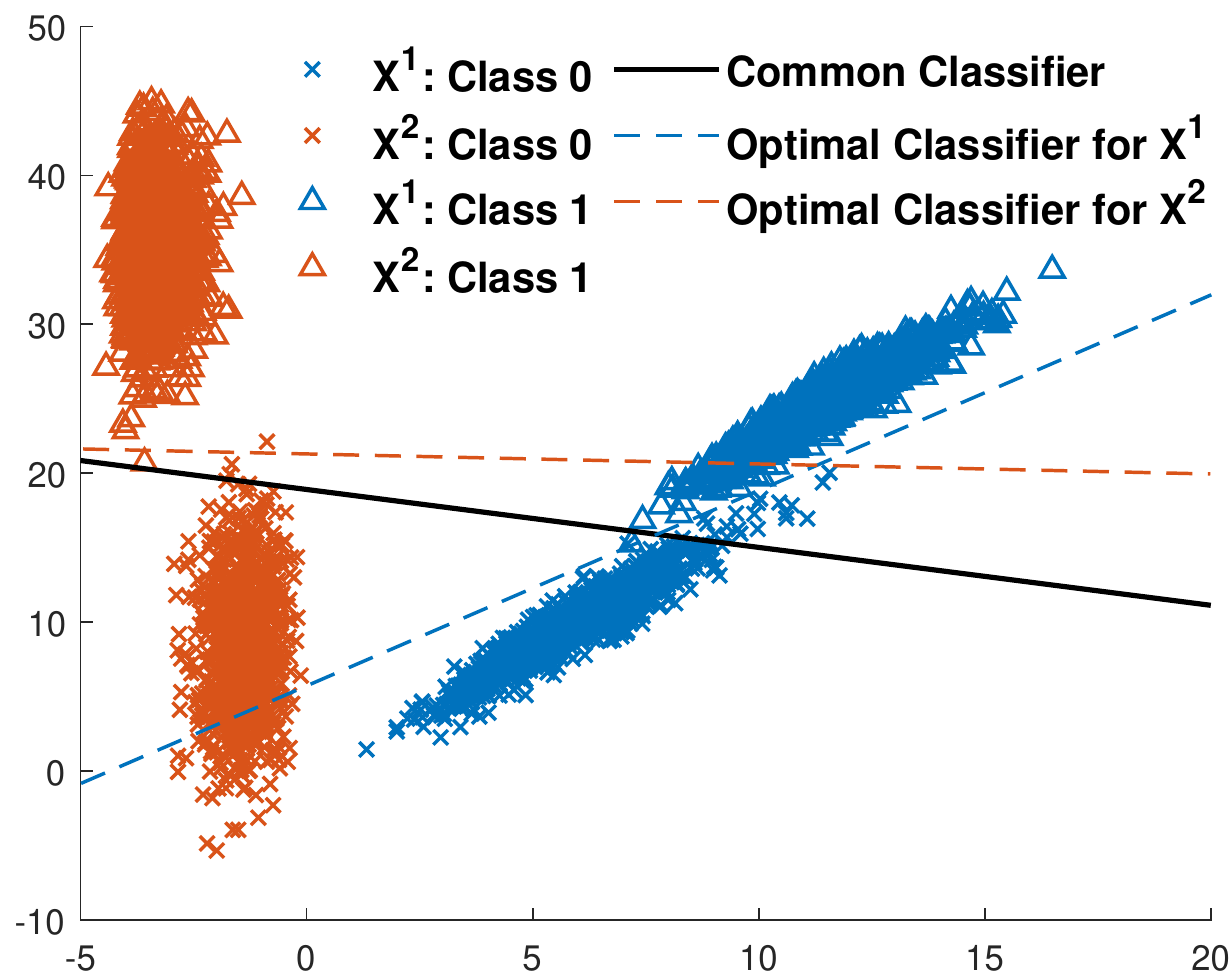}} \hfill
	\subfloat[Transformations $\bm{Z^1U^1X^1}$ and $\bm{Z^2U^2X^2}$ as determined on enforcing commutation between $\bm{Z^1} \text{ and } \bm{Z^2}$]{\includegraphics[width = 0.33\textwidth]{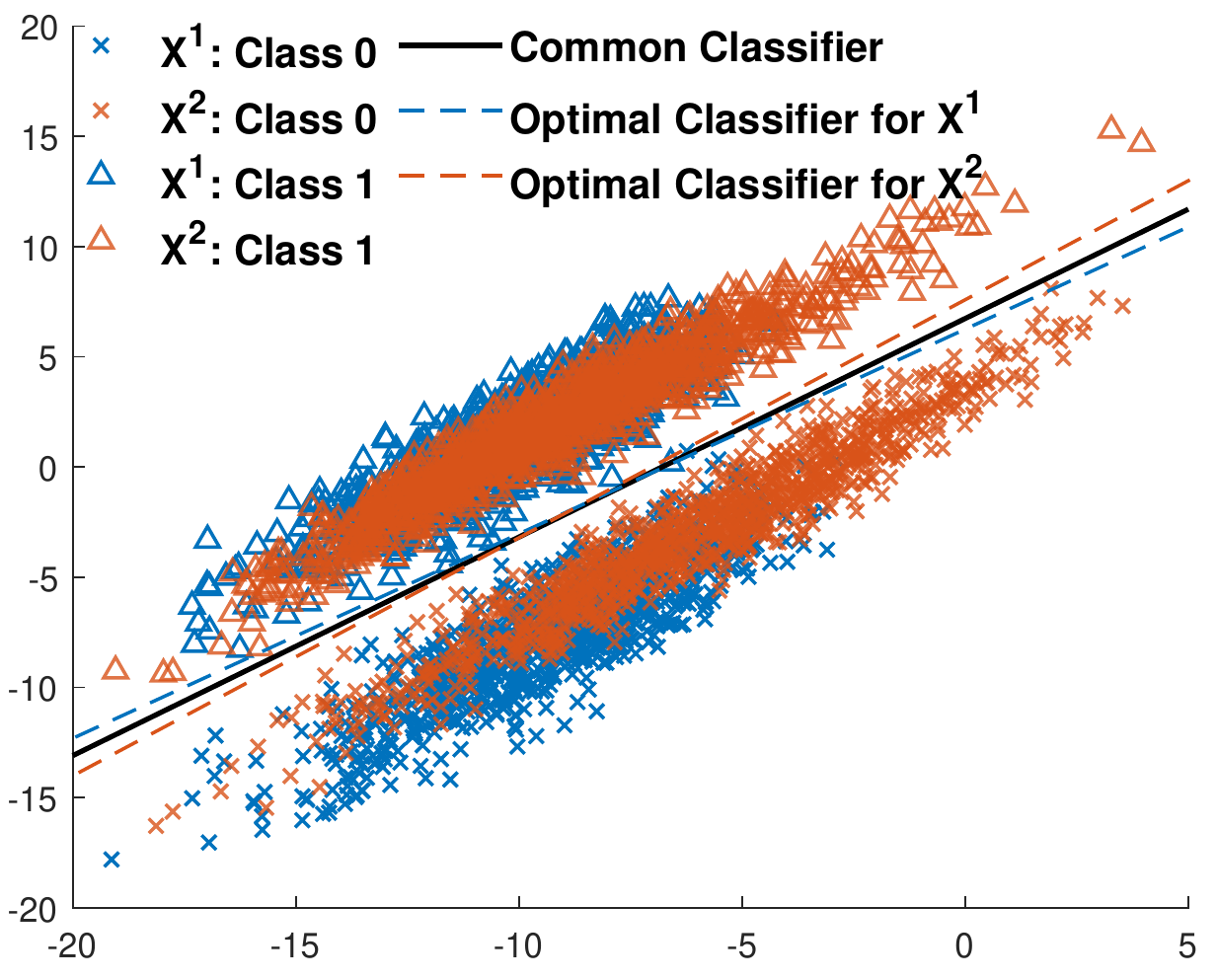}}
	\caption{}
	\label{toy_eg}
\end{figure}   

\section{Experiments and Results} \label{exp}
To substantiate the above proposed approach in the various scenarios, we select two different datasets.
\subsection{Dataset 1: Radar and Telescopic Imaging Sensors}
For the first dataset, we select two sensors, namely a Radar sensor and a telescopic optical sensor, the latter having been measured and collected by Jen-Hung Wang and the TAOS team. Due to technical difficulty in the field experiment, radar measurements were simulated according to the physical data of the space debris and matched with the optical data. Both sensors are ideally synchronized when observing a given target, which in our case, is a space object as just noted. The radar simulations (obtained through MATLAB Simulink, whose block diagram can be seen in Figure \ref{MATLAB_Sim}), together with telescopic image data are used in our first experiment. Each generated radar signal over one second is correlated with two telescopic images. Samples of objects with different velocities, cross-sections, ranges, and aspect-ratios are generated. The radar signals are used to make decisions over velocity, range, and the cross-section, while the telescopic images are used to make decisions over the aspect-ratio, and displacement over time of an object in view.

	\subsubsection{Experiment Design} \label{impl1}
\begin{figure}[h!] 
	\centering
	\includegraphics[width = 0.9\textwidth]{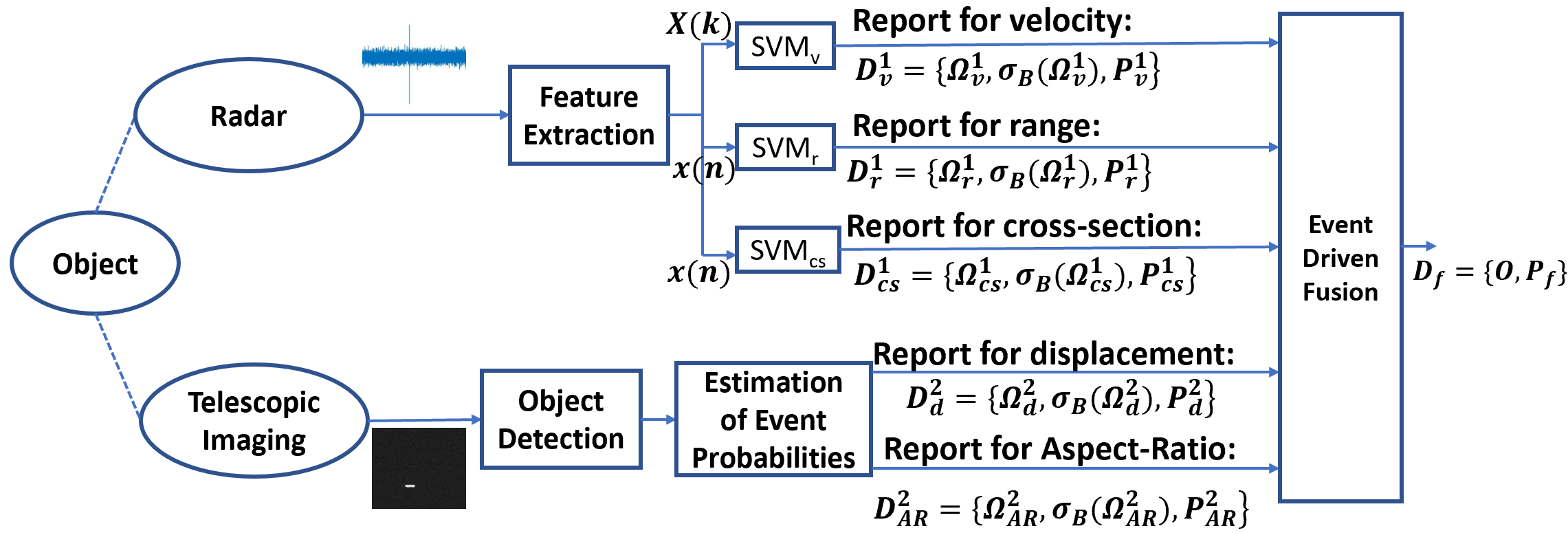}
	%\subfloat[]{\includegraphics[width = 0.5\textwidth, height= 0.17\textheight]{Space_BD_IHS.pdf}}
	%\centering
	\caption{Event Driven Fusion for object detection in Dataset 1} %, (b): Using Independent Hidden Spaces to deal with a damaged radar sensor}
	\label{EDF_Space_BD}
\end{figure}
\begin{figure*}[tbp] 
	\centering
	\includegraphics[width = 0.93\textwidth, height= 0.16\textheight]{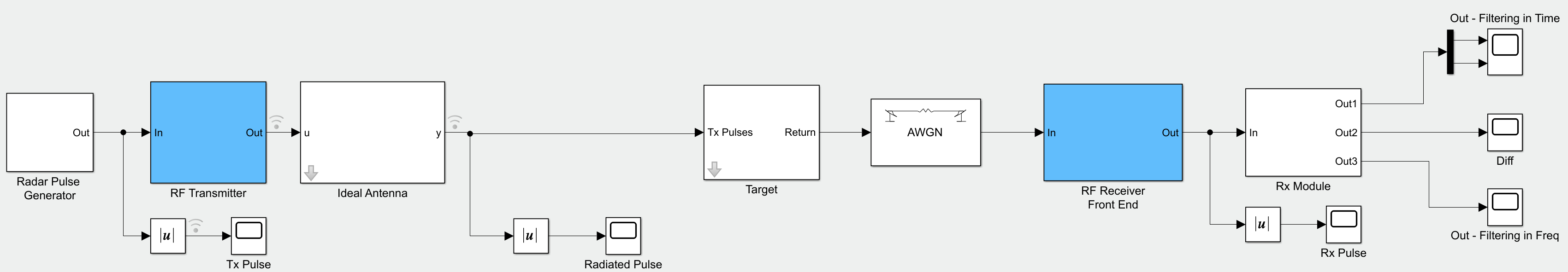}
	%\subfloat[]{\includegraphics[width = 0.5\textwidth, height= 0.17\textheight]{Space_BD_IHS.pdf}}
	%\centering
	\caption{MATLAB Simulink Block Diagram for generation of simulated radar signals} %, (b): Using Independent Hidden Spaces to deal with a damaged radar sensor}
	\label{MATLAB_Sim}
\end{figure*}
To proceed with the algorithmic evaluation, we must first generate distributions of features needed for target specification. Figure \ref{EDF_Space_BD} shows a high level block diagram for implementing an Event Driven Fusion for dataset 1. 
Let the received radar signal be, $x(n)$, and its corresponding Fourier transform, $ X(k) = \sum_{n=0}^{N-1} x(n)e^{(-i 2 \pi k n/N)} $ for each object, with associated labels distinguished by the object velocity, range, and cross-section values of that object, $ [v, r, cs] $. Using the training data, and the corresponding labels, SVM classifiers are trained over the events of interest defined over $[v,r,cs]$, and used to determine the classification probabilities for the event of interest, $P_k^l(\omega_{kj}^l)$, as described in Equation \ref{Pkl}. For the $k^{th}$ feature from $l^{th}$ sensor we train the SVM classifier using the Crammer-Singer formulation for multi-class SVM \cite{crammer-singer},
\begin{gather}
	\begin{align}
		&\nonumber \min_{\bm{W_k^l}, \bm{\xi_k^l}} \frac{1}{2} ||\bm{W_k^l}||^2 + C\sum_{n=1}^N \xi_{k_n}^l,\\
		&\nonumber \text{subject to: }\\
		&\nonumber w_{k_{y_n}}^{l^T}(\bm{x_n^l}) - w_{k_t}^{l^T} (\bm{x_n^l}) \geq 1 - \xi_{k_n}^l, \quad t \in \{1,...,J_{K_l}\},\\ 
		&\xi_{k_n}^l \geq 0,
	\end{align}
\end{gather}
where $y_n \in \{1,...,J_{K_l}\}$ is the true label of the $n^{th}$ data sample, and $N$ is the total number of data samples available for training.

We have two telescopic images associated with 1-sec of radar return for the same object. The object of interest in the telescopic image is first detected using target detection as discussed in \cite{bian}. Upon its detection, the probability distribution over the object's aspect-ratio, and its displacement in the second image relative to its location in the first image is determined using the image flow technique discussed in \cite{singh1992image}.
%\begin{comment}
\paragraph{Object Detection: } The object of interest in the telescopic imagery is initially detected by using target detection as discussed in \cite{bian}. Any pixel in the telescopic image domain is said to follow the probabilistic model, $I(o) = i_o + n, n \sim G(\mu_n, \sigma_n^2)$, where $o$ is the object that the pixel belongs to. Points in the image domain are said to belong to one of three sets in regards to the statistics of their neighborhoods\cite{bian}: 
	\begin{itemize}
	\item \textbf{Background Set:} For points in the image domain, the set of all background points is given as, 
	\begin{equation}
		W = \{p \in W|\forall q \in N(p), f(p) = f(q) = G(\mu_n,\sigma_n^2)  \},
	\end{equation}
	where, $f$ is the probability density function, and $\mathcal{N}(p)$ is the neighborhood of point $p$.
	\item \textbf{Interior Set:} The set of all interior points (in regards to objects) is defined as,
	\begin{equation}
		S = \{p \in S|\forall q \in N(p), f(p) = f(q) \neq G(\mu_n, \sigma_n^2)  \},
	\end{equation}
	where, $f$ is the probability density function, and $\mathcal{N}(p)$ is the neighborhood of point $p$.
	\item \textbf{Boundary Set:} The set of all boundary points is defined as,
	\begin{equation}
		B = \{p \in B|\exists q \in N(p), f(p)\neq f(q), f(p)\neq G(\mu_n, \sigma_n^2)  \},
	\end{equation}
	where, $f$ is the probability density function, and $\mathcal{N}(p)$ is the neighborhood of point $p$.
\end{itemize}
Based on the above definitions, a hypothesis test is performed in order to find the interior set of a given image. Let, $f_n = G(\mu_n, \sigma_n^2)$ be the background distribution, and $f_p = G(\mu_p, \sigma_p^2)$ be the distribution of object pixels, then,
\begin{gather}
	\nonumber H_0: p \sim f_p, q \sim f_p, \forall q \in N(p),\\
	H_1: p \sim f_n \text{ or } q \sim f_q \neq f_p, \forall q \in N(p). 
\end{gather}
After recognizing the interior points, the next crucial step is to cluster the interior points with respect to objects, for which a proper distance measure is important. The following metric, which reflects both the physical properties, such as, the apparent magnitude and  spatial relations is defined in \cite{bian},
\begin{gather}
	\begin{align}
		&\nonumber \forall p_a, p_b \in S,\\
		&\nonumber d(p_a, p_b) = d_{\text{Euclidean}} (p_a,p_b) + d_{\text{Intensity}} (p_a,p_b)\\
		&\quad \quad \quad \quad = \sqrt{(p_a^x - p_b^x)^2 + (p_a^y - p_b^y)^2} + \beta |I(p_a) - I(p_b)|,
	\end{align}
\end{gather}
where, $\beta$ balances the contribution of intensity distance and euclidean distance. The corresponding distance matrix is then used as an input to a clustering algorithm. Using single linkage clustering, two sets of pixels, $A$ and $B$ are said to belong to the same cluster if,
\begin{equation}
	\min\{d(p_a, p_b): p_a \in A, p_b \in B\} < \gamma
\end{equation}
The cut-off distance $\gamma$ can be estimated from the training data, and depends on the expected size of objects.
Such a clustering algorithm also identifies night sky stars as objects, which may not necessarily be of interest. In order to identify objects of interest, the ratio of width and height of an object in the image domain is used as a criterion. The spread of pixels representing an object (or the aspect-ratio of the object) is defined as $R(o) = W(o)/H(o)$, then, $o$ is an object of interest if, 
\begin{equation}
	|R(o)-c| > median\{R\},
\end{equation} 
\begin{figure*}[tbp]
	\subfloat[ ]{\includegraphics[width = 0.33\textwidth]{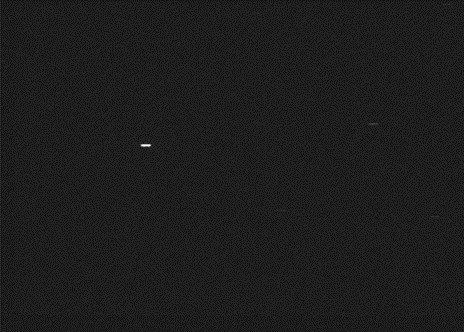}} \hfill 
	%\subfloat[ ]{\includegraphics[width = 0.33\textwidth, height=0.15\textheight]{TI_noise_removed_sq.jpg}} \hfill
	%\subfloat[ ]{\includegraphics[width = 0.33\textwidth, height=0.15\textheight]{Clustering_TI_sq.png}} \hfill
	\subfloat[ ]{\includegraphics[width = 0.33\textwidth]{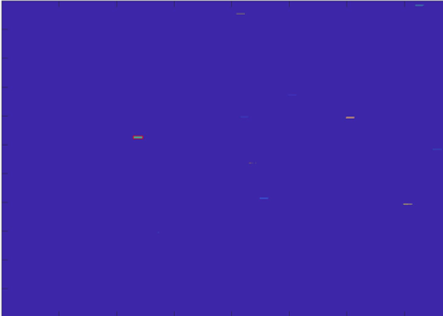}} \hfill
	\subfloat[ ]{\includegraphics[width = 0.33\textwidth, height=0.15\textheight]{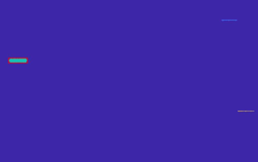}} 
	\caption{(a): Original Telescopic Image, (b): Detected Object of interest is marked by a red outline, (c): Detected Object of interest is marked by a red outline(zoomed in)}
	\label{clustering}
\end{figure*}
where, $c$ is the parameter measuring the distance between an object of interest and stars, and $median\{R\}$ is the median of width-height ration for all objects.
%\end{comment}  
Figure \ref{clustering} shows the clustering results, with a subsequent detection of object of interest for a sample telescopic image. Although not visible to the naked eye in Figure \ref{clustering}-(a), there are stars in the background, which are detected by the clustering algorithm, and can be seen in Figure \ref{clustering}-(b),(c).
\paragraph{Displacement Estimation} For two successive images, $I_1$ and $I_2$, captured by the telescopic sensor, a point $P(x,y)$ in $I_1$ moves to $P(x+u, y+v)$ in $I_2$, for which the displacement vector, $(u,v)$ is of interest. A correlation window of size, $\{max[W(o), H(o)] \times max[W(o), H(o)]\}$ is defined about the centroid of the target of interest in $I_1$. An error distribution is subsequently computed over a search window $(I_2)$ by using sum of squared distances,
\begin{equation}
	E(\mathcal{K},\mathcal{L}) = \sum_{i,j = -N}^N [I_1(x+i, y+j) - I_2(\mathcal{K}+i, \mathcal{L}+j)]^2,
\end{equation}
where, $N = \frac{max[W(o), H(o)]}{2}$, $0 < \mathcal{K} < W(I_2)$, and $0 < \mathcal{L} < H(I_2)$. This error distribution can then be converted into a probability distribution as, 
\begin{equation}
	P_d(\mathcal{K},\mathcal{L}) = e^{(-E(\mathcal{K},\mathcal{L})/z)},
\end{equation}
where, $z$ is a scaling factor. Furthermore, given a position $(x,y)$ of the object in image $I_1$, we get the probability distribution over $(u,v)$, $P_d(u,v) = e^{(-E(x+u, y+v)/z)}$. The probability of an event over displacement of the object can then be determined as,
\begin{equation}
	p(a<d<b) = \sum_{u,v} \mathcal{I}(a<d<b).P_d(u,v),
\end{equation}
where, $d = \sqrt{u^2 + v^2}$, and $\mathcal{I}$ is the indicator function. Figure \ref{pd} shows the estimation of this probability distribution for a sequence of two images.
\begin{figure}[!h]
	\centering
	\subfloat[Image 1 – Location of object of interest: $(x,y) = (429, 932)$, Image 2 – Location of object of interest: $(x+u, y+v) = (503, 932)$]{\includegraphics[width = 0.6\textwidth]{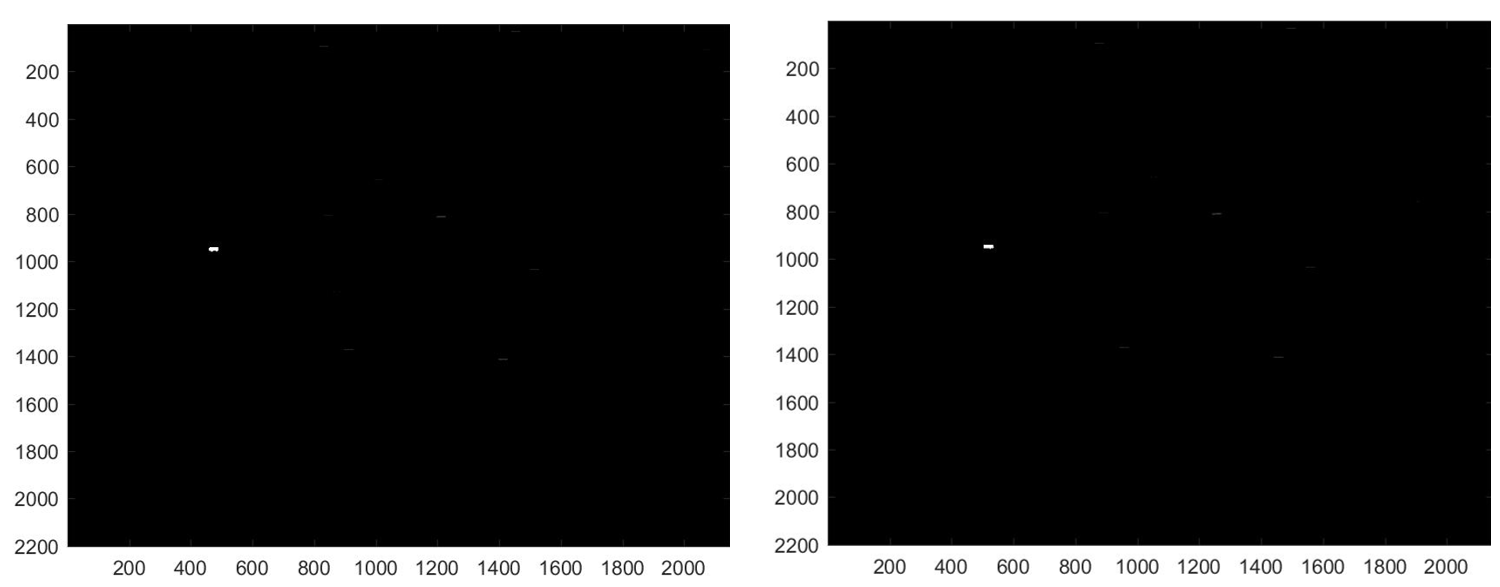}} \hfill
	\subfloat[ max probability located at: $(x+u, y+v) = (502, 932)$ ]{\includegraphics[width = 0.6\textwidth]{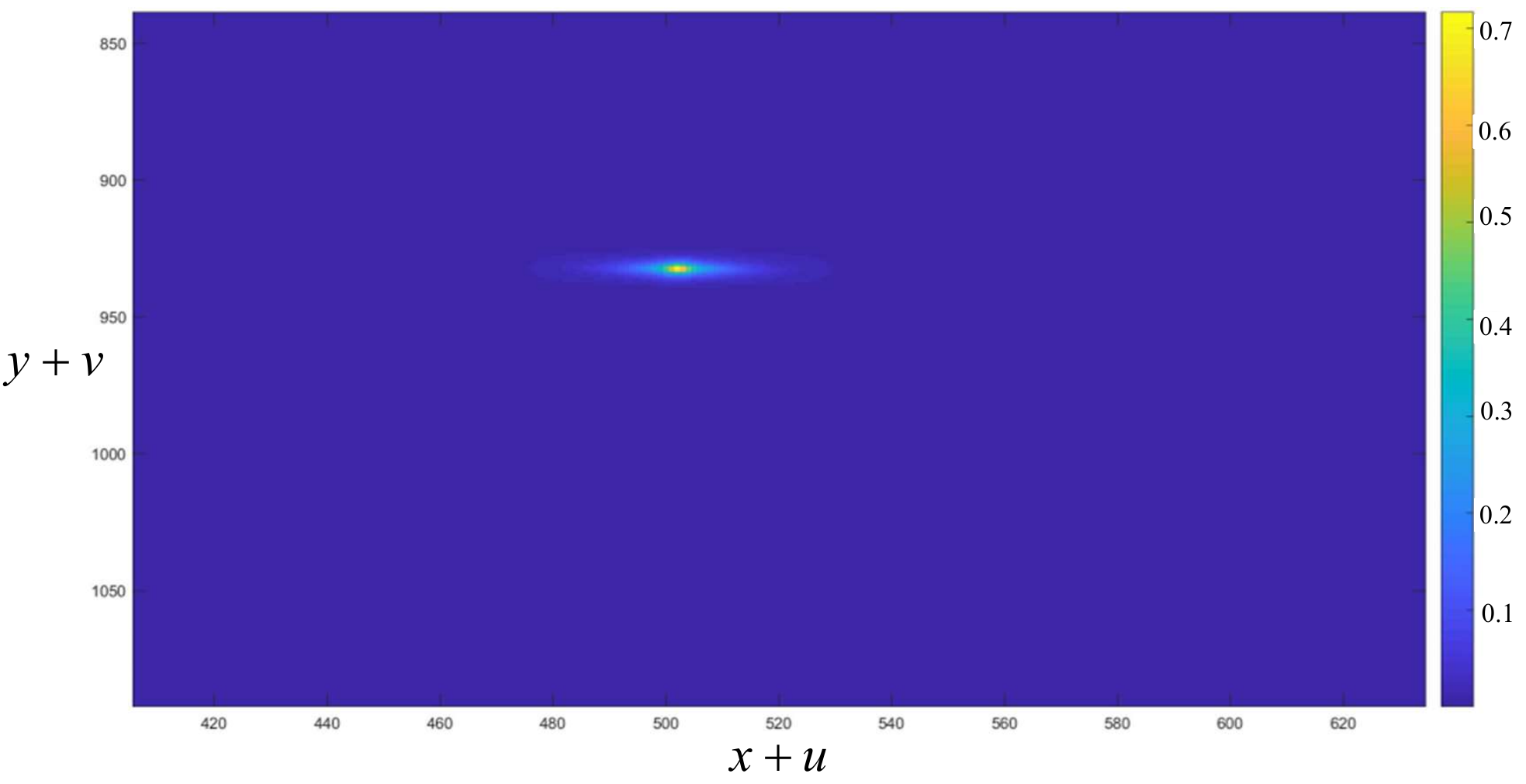}} \hfill
	\caption{(a): Sequence of two consecutive images from telescopic sensor, (b): Probability distribution over location of object of interest in $I_2$}
	\label{pd}
\end{figure}

\subsubsection{Object and Event Definitions}
For training and testing purposes, we define various events over the feature-sets from both sensors. Since real radar observations were not available for the first dataset, we have simulated two types of objects. The first is a ‘dangerous’ object defined as being close to the sensors, having a large cross section or aspect ratio or has a high velocity. On the other hand, the second object is defined as something ‘safe’ hence far from the sensors, has a small cross section and aspect ratio, and also a small velocity. Note that all operators in the definition of a ‘safe’ object are ‘and’ operators since it is important for the object to satisfy all of these constraints to be classified as a safe object. On the other hand, there is a ‘or’ between velocity and aspect ratio in the ‘dangerous’ object as a small object with a large velocity can also be dangerous and vice versa. These objects have been defined by the authors as a way of evaluating the model but there may be other definitions and hence are problem-dependent.

For the radar, as noted before, we use $ [v, r, cs] $ and the events are defined as,
\begin{gather}
	\nonumber a_1^v: 0 \leq v \leq 10 \text{ }mi/s,\text{ } a_2^v: 15\text{ }mi/s \leq v \leq 35\text{ }mi/s,\\
	\nonumber a_1^r: 0 < r \leq 300 \text{ }mi,\text{ }a_2^r: 300\text{ }mi < r, \\
	a_1^{cs}: 0 < cs \leq 20 \text{ }m^2, \text{ }a_2^{cs}: 15\text{ }m^2 \leq cs \leq 50\text{ }m^2.
\end{gather}
From the telescopic imaging sensor, the features, displacement and aspect ratio, $[d, AR]$ define the following events,
\begin{gather}
	\nonumber a_1^d: 0 \leq d \leq 60 \text{ }pixels,\text{ }a_2^d: 90\text{ }pixels \leq d \leq 210\text{ }pixels,\\
	a_1^{ar}: 0 < AR \leq 1.5, \text{ }a_2^{ar}: 1.5 < AR.
\end{gather}
Furthermore, the objects for classification are defined in terms of these events as,
\begin{gather}
	o_1 (\text{dangerous object}) : \{a_1^r \wedge [(a_2^v \wedge a_2^d) \vee (a_2^{cs} \vee a_2^{ar})]\},\\
	o_2 (\text{safe object}) : \{a_1^v \wedge a_1^d \wedge a_2^r \wedge a_1^{cs} \wedge a_1^{ar}\}.
\end{gather}
Given these events and object definitions, we determine the fused report, $ D_f = \{P_f(o_1), P_f(o_2), P_f(\overline {{o_1} \vee {o_2}})\} $ using our proposed approach. This can be considered a classification problem with 3 classes, Class 1:Object 1, Class 2:Object 2, and Class 3:Neither Object 1 nor Object 2.

\subsection{Dataset 2: Acoustic and Seismic Sensors}
\begin{figure}[h!]
	\centering
	\subfloat[]{\includegraphics[width=0.49\textwidth]{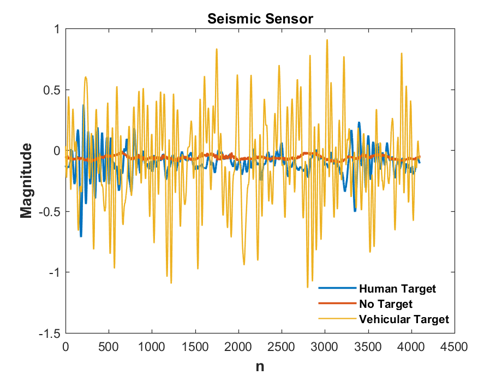}}\hfill
	\subfloat[]{\includegraphics[width=0.49\textwidth]{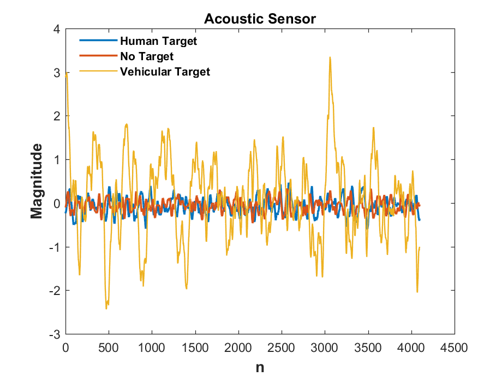}}\hfill
	\caption{Sample (a): seismic sensor observations and (b): acoustic sensor observations for human, vehicular, and no target cases}
	\label{Sample_Obs}
\end{figure}
\begin{figure}
	\centering
	\includegraphics[width=0.95\textwidth]{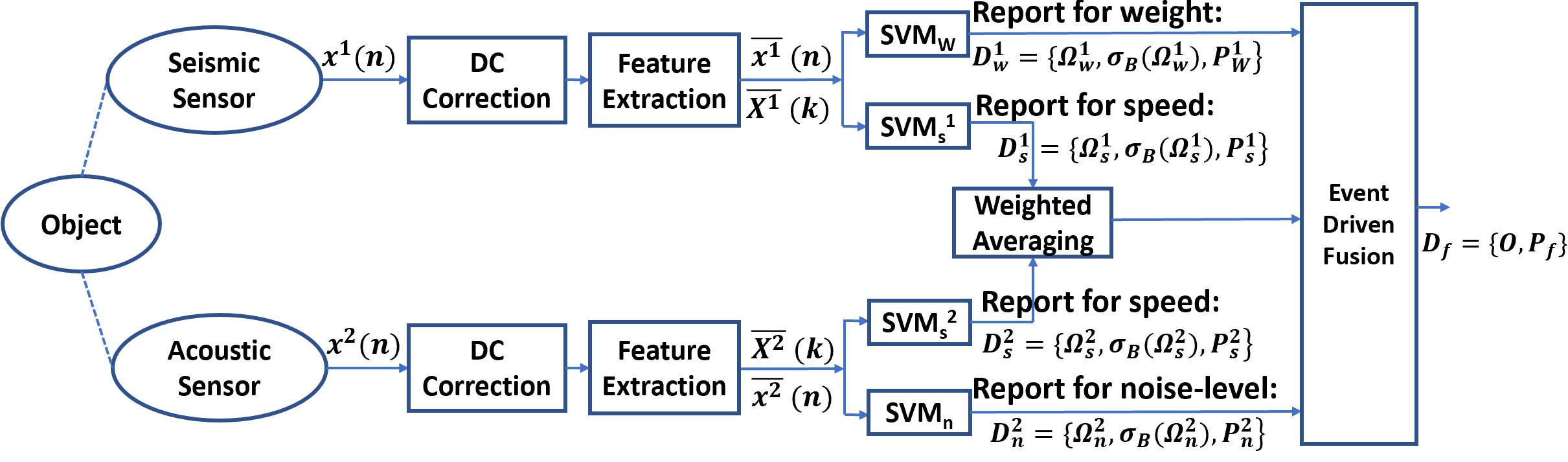}
	\caption{Event Driven Fusion for object detection - Dataset 2}
	\label{algo2}
\end{figure}

A second dataset we use in our experimental validation; is pre-collected data from a network of seismic sensors, and acoustic sensors deployed in a field, where people/vehicles were walking/driven around in specified patterns. Details about this sensor setup and experiments can be found in \cite{nabritt2015personnel}. This dataset has been previously used for target detection in \cite{lee2017accumulative, roheda2018cross}, where, the authors focused on detection of human targets. Here, we use this dataset to classify between human targets, vehicular targets, and no targets. Some data samples from the sensors are shown in Figure \ref{Sample_Obs}

\subsubsection{Experiment Design}
Figure \ref{algo2} shows a high level block diagram for the implementation of Event Driven Fusion for the second dataset. %Let the seismic and acoustic signals be, $x(n)$, and $z(n)$, respectively. First the DC offset is removed from the signal to get the corrected signals, $\overline{x}(n)$, and, $\overline{z}(n)$. Further, the corresponding Fourier transforms are computed to get $\overline{X}(k)$, and, $\overline{Z}(k)$. 
Using the training data, SVM classifiers are trained over the corresponding events of interest, as discussed before for the first dataset in Section \ref{impl1}. The seismic sensor provides decisions over the features, target weight and target speed, $[w, s]$. True labels for target weights are provided in the dataset, while those for target speeds are obtained from the GPS data of the target. Similarly, the acoustic sensor provides decisions on the noise-level of the target, and the target speed, $[n, s]$. The two decisions over the target speed are combined into a single report by performing a weighted averaging of the decisions of the two sensors. Here, the weights are selected on the basis of the individual accuracies of the SVMs trained to detect events on target speed.
\subsubsection{Object and Event Definitions}
For training and testing purposes, we define various events over the feature-sets from both the sensors. For the seismic sensor, we use $[w, s]$, while for the acoustic sensor we use $[n, s]$.
\begin{gather}
	\nonumber a_1^w: 96.08 \text{ pounds} \leq w \leq 230.61 \text{ pounds},\\
	\nonumber a_2^w: 1311.61\text{ pounds} \leq w,\\
	\nonumber a_1^s: 0.37 \text{ m/s} < s \leq 2.12 \text{ m/s}, \text{ }a_2^s: 1.7\text{ m/s} \leq s,\\
	a_1^n: n \leq -30 \text{ db},\text{ }a_2^n: -10.6658 db \leq n \leq 7.84 db.
\end{gather}
The range of an event can be determined from the training data. The mean of the feature in question over the samples of the same class is computed, and a range of twice the standard deviation is taken on either side of the mean. 
Furthermore, the targets are defined as,
\begin{gather}
	o_1\text{ }(human\text{ }target): \{a_1^s \wedge (a_1^w \vee a_1^n)\},\\
	o_2\text{ }(vehicular\text{ }target): \{a_2^s \wedge a_2^w \wedge a_2^n\}.
\end{gather}
%The definition of human target is chosen as such since the noise interference due to wind is very high, and the acoustic sensor does not do well in discriminating between the no-target class and human target class. 
Given these events and object definitions, we wish to determine the fused report, $ D_f = \{P_f(o_1), P_f(o_2), P_f(\overline {{o_1} \vee {o_2}})\} $, where, $\{\overline{{o_1} \vee {o_2}}\}$ represents the no target case.

\subsection{Performance Analysis}
Table \ref{acctable}, and \ref{acctable2} show the  classification performance of different techniques (averaged over 10 runs of the technique) when implemented on dataset 1 and 2 respectively.
\begin{table}[htbp]
	\caption{Performance Comparison for the First Dataset}
	\label{acctable}
	\begin{center}
		\begin{tabular}{|c|c|c|c|}
			\hline
			\textbf{Method} & \textbf{Average Accuracy}\\
			\hline
			Radar & 86.47\%\\
			\hline
			Telescopic Imaging & 81.31\%\\
			\hline
			Feature Concatenation & 85.93\%\\
			\hline
			Similar Sensor Fusion  & 86.07\%\\ %Convex Quadratic Fusion
			\hline
			Dissimilar Sensor Fusion  & 88.61\%\\ %Analytic Center Fusion
			\hline
			Dempster-Shafer Fusion & 87.18\%\\
			\hline
			\textbf{Event Driven Fusion} & \textbf{90.36\%}\\
			\hline
		\end{tabular}
	\end{center}
\end{table}

\begin{table}[htbp]
	\caption{Performance Comparison for the Second Dataset}
	\begin{center}
		\begin{tabular}{|c|c|c|c|}
			\hline
			\textbf{Method} & \textbf{Average Accuracy}\\
			\hline
			Seismic Sensor & 85.41\%\\
			\hline
			Acoustic Sensor & 67.62\%\\  %77.62
			\hline
			Feature Concatenation & 81.63\%\\
			\hline
			Similar Sensor Fusion  & 86.69\%\\ %Convex Quadratic Fusion 
			\hline
			Dissimilar Sensor Fusion & 89.96\%\\ %Analytic Center Fusion
			\hline
			Dempster-Shafer Fusion & 87.93\%\\
			\hline
			\textbf{Event Driven Fusion} & \textbf{92.04\%}\\
			\hline
		\end{tabular}
		\label{acctable2}
	\end{center}
\end{table}

\begin{figure*}[tbp]
	\centering
	\subfloat[ ]{\includegraphics[width = 0.3\textwidth, height=0.2\textheight]{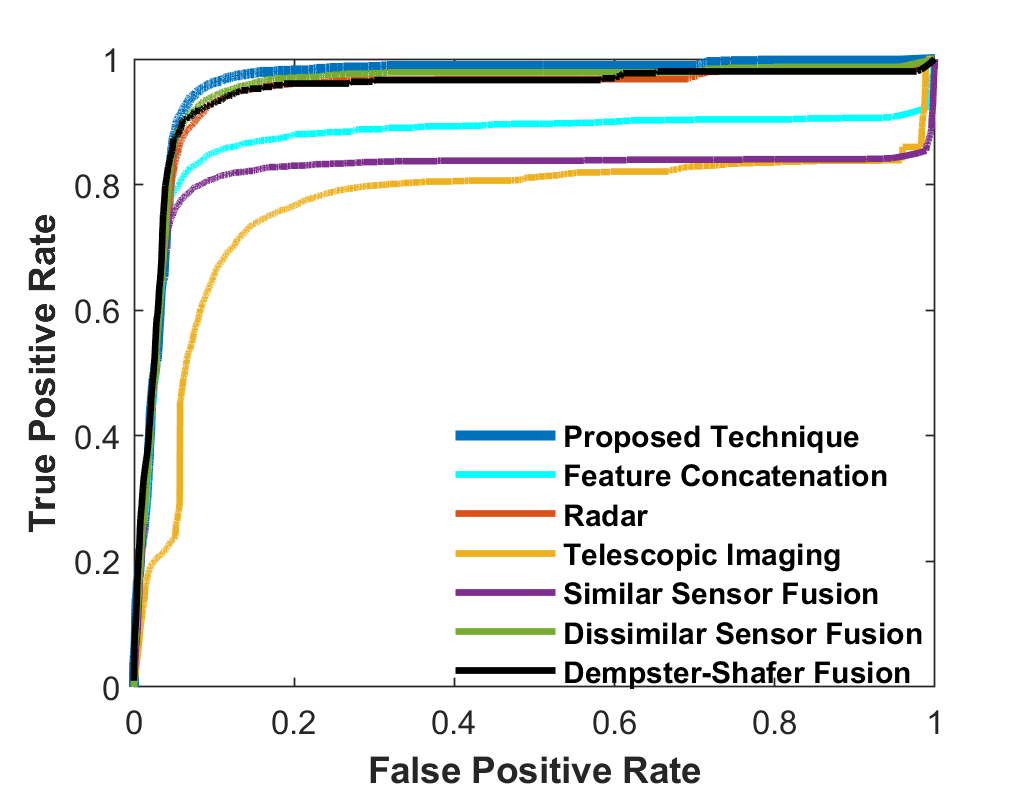}} \hfill 
	\subfloat[ ]{\includegraphics[width = 0.3\textwidth,height=0.2\textheight]{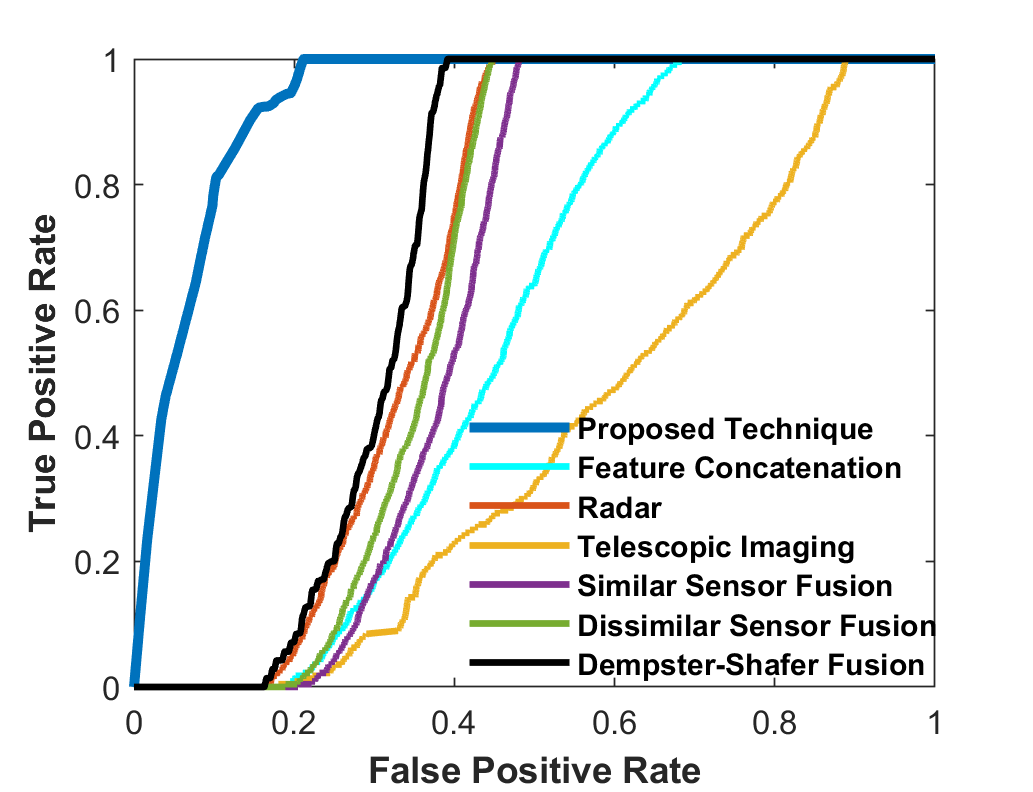}} \hfill
	\subfloat[ ]{\includegraphics[width = 0.3\textwidth,height=0.2\textheight]{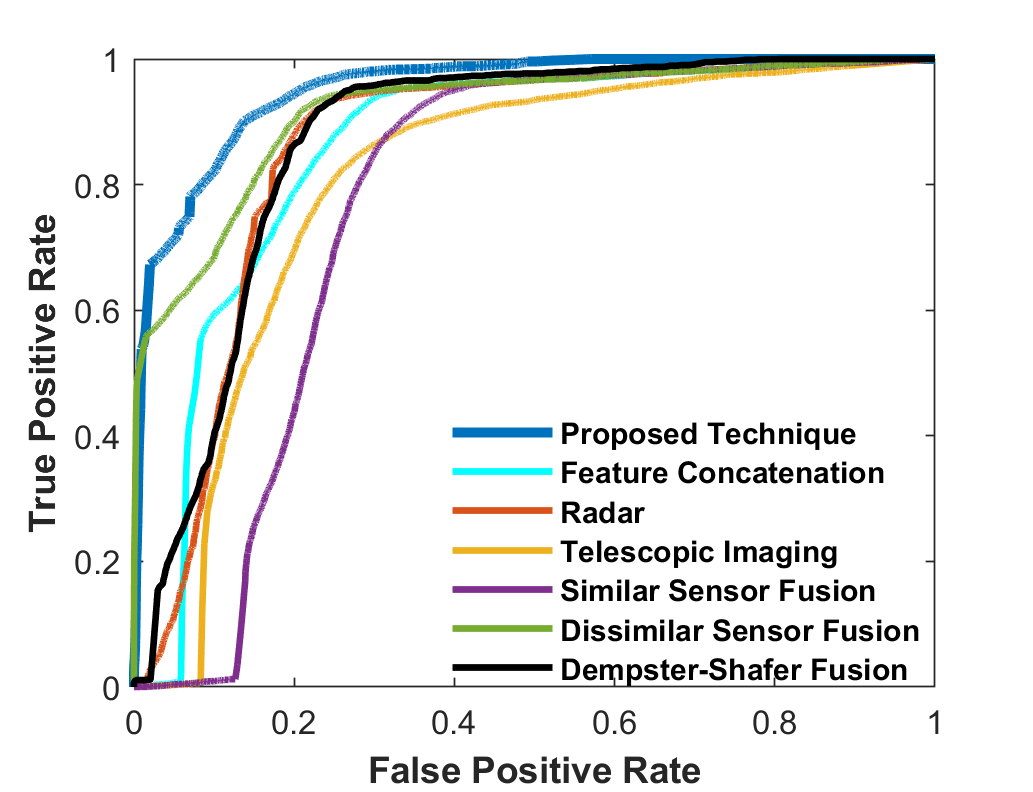}} \hfill
	\caption{\small ROC Curves for detection (Dataset 1) of (a): Class 1:Object 1, (b): Class 2:Object 2, (c): Class 3:Neither Object 1 nor Object 2}
	%edit this algo
	\label{roccurves}
\end{figure*}

Classification accuracy is often not the best measure to quantify performance, particularly in cases where different classes have different numbers of samples, which is the case here. A better way to compare performance is to look at the Receiver Operating Characteristic (ROC) curves. Fig. \ref{roccurves} and \ref{roccurves2} show the ROC curves for classification for each of datasets 1 and 2 respectively. The true positive rate is used to measure the fraction of samples that were correctly classified as positive. On the other hand, the false positive rate is used to measure the fraction of samples that were incorrectly classified as positive. In our case, true positives for class $k$ are those that were correctly classified as belonging to class $k$, and false positives are those that were classified into class $k$ but do not actually belong to class $k$.
It can be seen from the ROC curves (for dataset 1) in Fig. \ref{roccurves}, that other techniques show limited performance in correct classification of objects from class 2 due to the low number of samples for class 2 in comparison to those in class 1 and class 3. This causes the classifier to bias toward selecting class 1 or class 3 in order to achieve high classification accuracy (even when the sample is from class 2). But, our technique trains over occurrence of events rather than the object itself, hence does not face this issue. 
Improvement in performance is also seen for Dataset 2 (Fig. \ref{roccurves2}). In particular, detection of human targets is significantly improved, by taking 'or' between noise level event and weight event, which reduces misclassification due to noise due to winds.
\subsection{Robustness Evaluation}

\begin{figure*}[tbp] 
	\centering 
	\subfloat[ ]{\includegraphics[width = 0.33\textwidth,height=0.2\textheight]{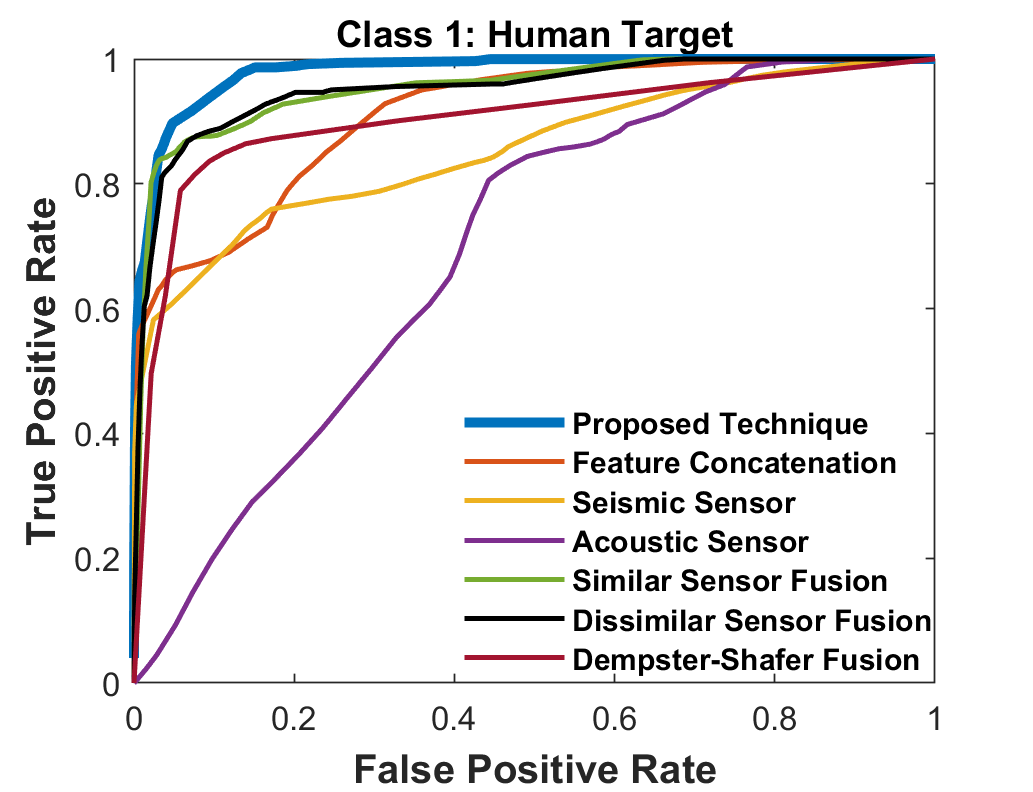}} \hfill
	\subfloat[ ]{\includegraphics[width = 0.33\textwidth,height=0.2\textheight]{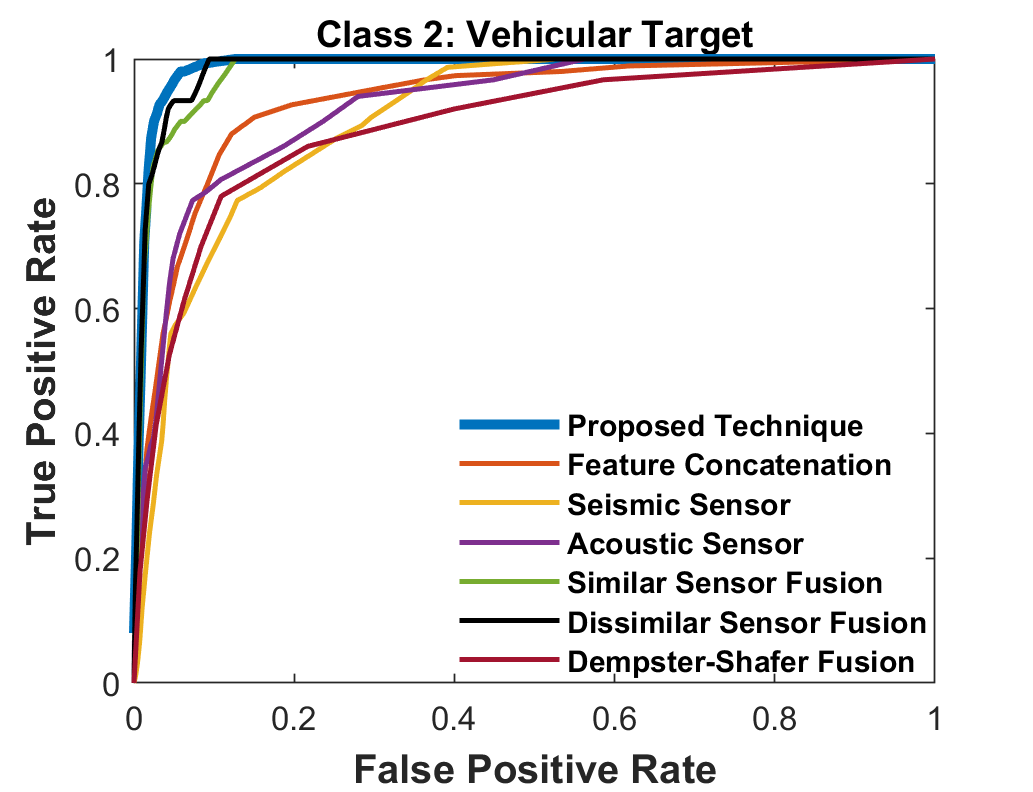}} \hfill
	\subfloat[ ]{\includegraphics[width = 0.33\textwidth, height=0.2\textheight]{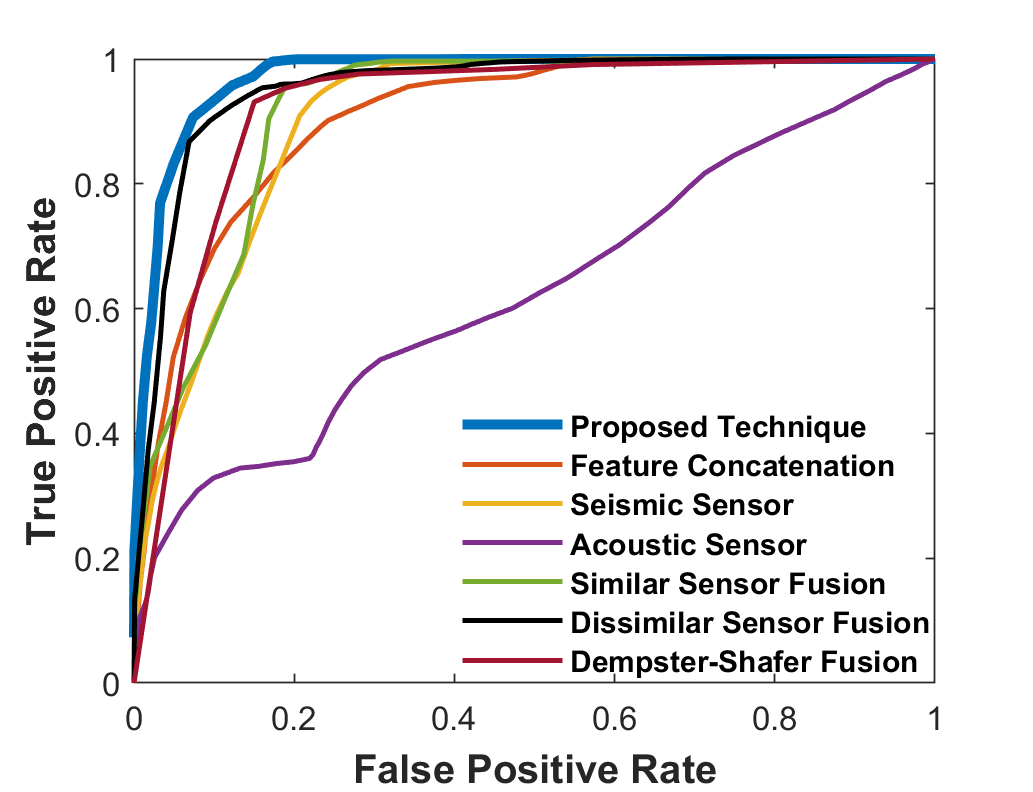}} \hfill
	\caption{\small ROC Curves for detection (second dataset) of (a): Class 1: Human Target, (b): Class 2: Vehicular Target, (c): Class 3: No Target}
	%edit this algo
	\label{roccurves2}
\end{figure*}

\begin{figure*}
	\centering
	%	\subfloat[ ]{\includegraphics[width = 0.48\textwidth, height= 0.17\textheight]{Space_BD_IHS.pdf}} \hfill
	\subfloat[ ]{\includegraphics[width = \textwidth]{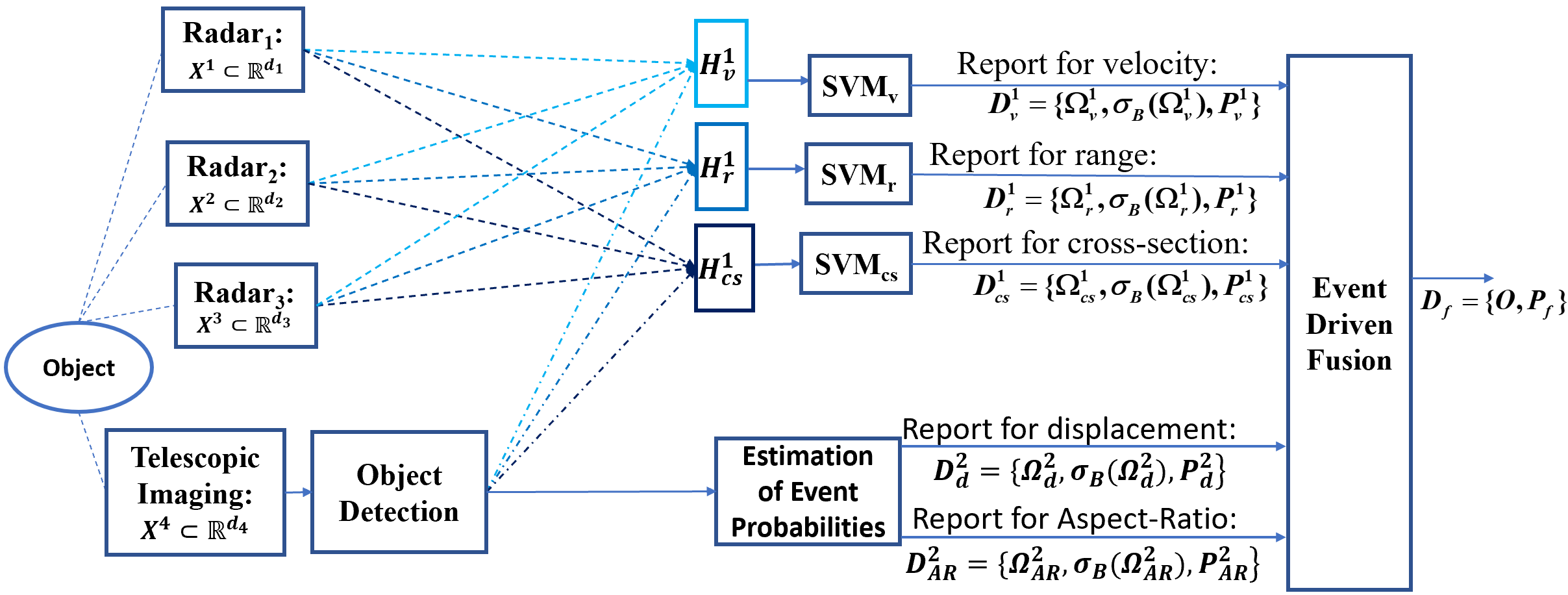}} \hfill
	\subfloat[ ]{\includegraphics[width = 0.85\textwidth]{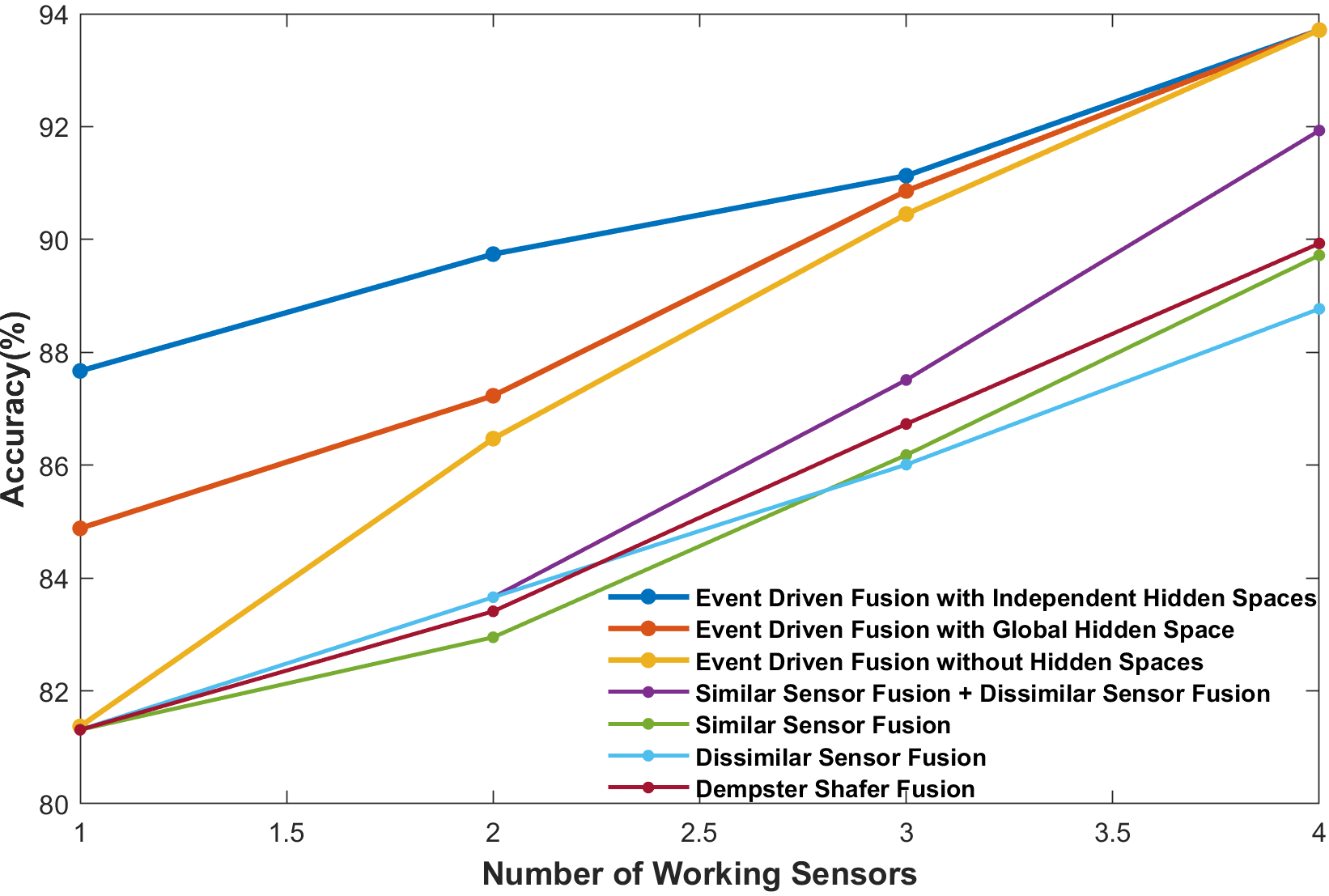}} \hfill
	\caption{(a): Using Independent Hidden Spaces to deal with a damaged radar sensor, (b): Comparison of proposed approach with existing techniques, when sensors get damaged}
	\label{IHS_1}
\end{figure*}

In order to evaluate the robustness of the proposed algorithm, we consider a damaged sensor scenario that was discussed in Section \ref{damaged_sensors}. For the first dataset, we consider the situation where 3 radar sensors (simulated with different Signal to Noise Ratios) are used along with a telescopic sensor, and a subset of the radar sensors are damaged at a given time (Figure \ref{IHS_1}-(a)). The Global Hidden Space, and Independent Hidden Spaces approaches are evaluated and compared with the common approach of ignoring the damaged sensors in Figure \ref{IHS_1}-(b). The `Similar Sensor Fusion + Dissimilar Sensor Fusion' case in Figure \ref{IHS_1}-(b) refers to the fusion of the radar sensors using Similar Sensor Fusion, followed by fusion with telescopic sensor using Dissimilar Sensor Fusion.  As the number of working sensors is reduced, the target detection performance is impacted. We note that the exploitation of Independent Hidden Spaces along with Event Driven Fusion allows for a more graceful degradation. In the second dataset case, we consider the following scenarios, 1) Seismic Sensor is damaged 2) Acoustic Sensor is damaged (Figure \ref{IHS_2}). The results for these two cases are shown in Table \ref{acctable_damaged}. A similar observation can be made here as in spite of one of the damaged sensors during testing, the prior information from the training phase allows us to learn a transform that helps boost the performance of the working sensor, hence gracefully mitigating the impact.
\begin{table}
	\caption{Robustness Analysis for the second dataset}
	\label{acctable_damaged}
	\begin{center}
		\begin{tabular}{|M{6cm}|M{1.5cm}|}
			\hline
			\textbf{Method} & \textbf{Average Accuracy}\\
			\hline
			Seismic Sensor & 85.41\%\\
			\hline
			Acoustic Sensor & 67.62\%\\
			\hline
			Feature Concatenation & 81.63\%\\
			\hline
			Similar Sensor Fusion  & 86.69\%\\ %Convex Quadratic Fusion
			\hline
			Dissimilar Sensor Fusion  & 89.96\%\\ %Analytic Center Fusion
			\hline
			Dempster-Shafer Fusion & 87.93\%\\
			\hline
			\textbf{Event Driven Fusion} & \textbf{92.04\%}\\
			\hline 
			\small Event Driven Fusion with Global Hidden Space (Damaged Seismic Sensor) & \small68.33\%\\
			\hline 
			\small Event Driven Fusion with Global Hidden Space (Damaged Acoustic Sensor) & \small 86.13\%\\
			\hline
			\small \textbf{Event Driven Fusion with Independent Hidden Spaces (Damaged Seismic Sensor)} & \small \textbf{71.66\%} \\
			\hline
			\small \textbf{Event Driven Fusion with Independent Hidden Spaces (Damaged Acoustic Sensor)} & \small \textbf{87.36\%} \\
			\hline
		\end{tabular}
	\end{center}
\end{table}
\begin{figure}[h!]
	\centering
	\includegraphics[width=0.95\textwidth]{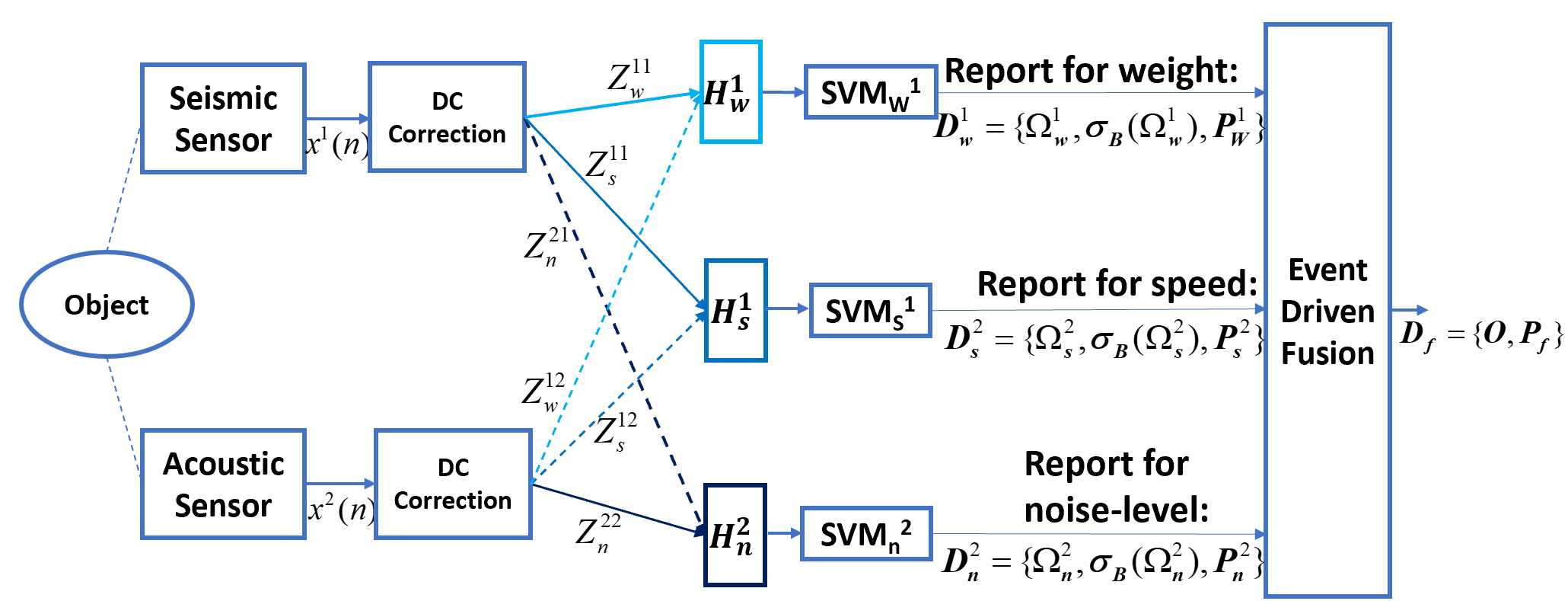}
	\caption{Using Independent Hidden Spaces to deal with damaged Seismic/Acoustic sensors}
	\label{IHS_2}
\end{figure}

\section{Conclusion}
We proposed a novel sensor fusion technique that looks at targets as combinations of probabilistic events defined over the feature set used to construct a sigma-field all the while considering the extent of dependence between different features. Experiments on various datasets showed that the proposed technique can outperform existing fusion techniques on the decision level. %particularly when data samples are unequally distributed across different targets/objects. 
We also propose a technique to safeguard detection performance of the model when sensors are damaged during the implementation phase by leveraging the prior information available during training. 





\end{document}